\documentclass{article}

\usepackage[margin=1in]{geometry}

\usepackage{microtype}
\usepackage{graphicx}
\usepackage{subfigure}
\usepackage{booktabs} 
\usepackage[numbers, sort, comma, square]{natbib}
\usepackage[colorlinks=true, allcolors=blue]{hyperref}

\usepackage{amsmath}
\usepackage{amssymb}
\usepackage{mathtools}
\usepackage{amsthm}
\usepackage{thm-restate}

\theoremstyle{plain}
\newtheorem{theorem}{Theorem}[section]

\newtheorem{lemma}[theorem]{Lemma}

\theoremstyle{definition}

\theoremstyle{remark}

\newtheorem{problem}{Problem}[section]

\usepackage[algo2e,linesnumbered,ruled,vlined]{algorithm2e}
\usepackage{bm}
\usepackage{bbm}
\usepackage{xcolor}
\usepackage{times}
\usepackage{url}

\newcommand{\E}{\mathbb{E}}

\newcommand{\OPT}{\mathrm{OPT}}
\newcommand{\R}{\mathbb{R}}

\newcommand\numberthis{\addtocounter{equation}{1}\tag{\theequation}}

\begin{document}

\title{An Asymptotically Optimal Approximation Algorithm\\ for Multiobjective Submodular Maximization at Scale}

\author{
  Fabian Spaeh \\ fspaeh@bu.edu \\ Boston University
\and
  Atsushi Miyauchi \\ atsushi.miyauchi@centai.eu \\ CENTAI Institute
}

\date{May 2025}

\maketitle

\begin{abstract}
  Maximizing a single submodular set function subject to a cardinality
  constraint is a well-studied and central topic in combinatorial optimization.
  However, finding a set that maximizes multiple functions at the same time is
  much less understood, even though it is a formulation which naturally occurs
  in robust maximization or problems with fairness considerations such as fair
  influence maximization or fair allocation.
    In this work, we consider the problem of maximizing the minimum over many
  submodular functions, which is known as multiobjective submodular
  maximization. All known polynomial-time approximation algorithms either
  obtain a weak approximation guarantee or rely on the evaluation of the
  multilinear extension. The latter is expensive to evaluate and renders such
  algorithms impractical. We bridge this gap and introduce the first scalable and
  practical algorithm that obtains the best-known approximation guarantee.
  We furthermore introduce a novel application \emph{fair centrality maximization}
  and show how it can be addressed via multiobjective submodular maximization.
  In our experimental evaluation, we show that our algorithm outperforms known
  algorithms in terms of objective value and running time.
\end{abstract}

\section{Introduction}\label{sec:intro}

Due to the natural diminishing returns property, maximizing submodular functions
is a central task in various fields such as optimization, machine learning, and economics~\cite{Fujishige05,Nemhauser+78}. 
In submodular maximization subject to a cardinality constraint 
the task is to find a set $S$ of size $|S| \le B$ for some budget $B$ such that a submodular objective $f(S)$ is maximized.
Applications range from welfare maximization to experimental design and viral
marketing via influence maximization \cite{Bilmes22, Kempe+03, Krause+07}.
Maximizing a single monotone submodular function is well understood, and the greedy algorithm is practically
efficient and obtains the best possible polynomial-time approximation guarantee of $1 - \frac 1 e \approx 0.63$~\cite{Nemhauser+78Best}.
Moreover, so-called \emph{lazy evaluations} make the greedy algorithm highly scalable~\cite{Minoux05}.

Another line of work studies the task where $k$ different submodular objectives
have to be simultaneously maximized, which is called multiobjective submodular maximization.
This is also important for many applications.
For instance, in robust experimental design, the goal is to maximize a submodular
function $f_\theta$, which depends on an a-priori unknown parameter $\theta$
\cite{Krause+08}. One seeks to find a set that maximizes $f_\theta$
simultaneously for all valid parameters~$\theta$.
Recently, many applications in artificial intelligence consider the max-min
fairness objective, which often naturally falls into the multiobjective scenario.
One such example is fair influence maximization, where the goal is to
maximize the influence among all groups~\cite{Tsang+19}.

Even though applications for the multiobjective setting are numerous and relevant,
solving the problem with theoretical guarantees still presents
challenges for current algorithms. Indeed, the prior work can be split into two
categories of approximation algorithms: First, there are algorithms that cannot
recover the $(1 - \frac 1 e)$-approximation, but are discrete and
efficient in practice. The best approximation ratio for such an algorithm
is $(1 - \frac 1 e)^2 \approx 0.40$ \cite{Udwani18}.
Second, there are algorithms that (almost) recover the
$(1 - \frac 1 e)$-approximation guarantee~\cite{Chekuri+10,Udwani18}, but solve a continuous relaxation
of the problem based on the multilinear extension, which is known to be
impractical~\cite{Bai+18, Chen+24, Buchbinder+24}. 
These algorithms do not scale to larger problem instances, as the ones
necessitated by the mentioned applications.

This is where our work comes in: We introduce an algorithm that (almost) achieves
the $(1 - \frac 1 e)$-approximation, but does not relax the problem to the
continuous setting and is therefore efficient in practice. This settles an open
question of \citet{Udwani18}, showing that there is a practically
efficient algorithm that obtains best possible approximation ratio.
Our main contributions are the following:
\begin{enumerate}
  \item We introduce a novel algorithm for multiobjective submodular
    maximization. Our algorithm (almost) achieves the best possible approximation ratio of $1 - \frac 1 e$. 
            Our algorithm is vastly more efficient than algorithms with comparable guarantees,
    which rely on a continuous relaxation and a rounding procedure.
    Our algorithm, instead,
    achieves its guarantee via a novel concentration argument.

  \item Furthermore, we show how multiplicative
    weights updates (MWU) can further reduce the number of function evaluations,
    which also yields an analogue to the lazy greedy algorithm for the multiobjective problem.

  \item We introduce a novel application \emph{fair centrality maximization}, 
    as a generalization of the well-studied task of centrality maximization with fairness consideration, 
    and show how it can be formulated as a large-scale multiobjective submodular maximization. 

  \item We experimentally verify that our algorithm outperforms previous methods
    in terms of objective value and running time. 
                The practical efficiency of our method allows us to be the first to solve large-scale instances 
    while retaining quality of solutions. \end{enumerate}

\section{Related Work}\label{sec:related}

\citet{Nemhauser+78} showed that the greedy algorithm achieves a
$(1 - \frac 1 e)$-approximation for cardinality-constrained monotone submodular
maximization. This is tight, unless
$\mathrm{P}\!=\!\mathrm{NP}$~\cite{Feige98}. More algorithms
for this and related settings exist
\cite{Calinescu+11,Badanidiyuru+14,Mirzasoleiman+15}.

\citet{Krause+08} introduced the multiobjective submodular
maximization problem in the context of robust experimental design.
They provided a bi-criteria algorithm and showed that the problem becomes
inapproximable in the regime where $B \le k$.
\citet{Chen+17} also proved that computing a $(1 - \frac 1 e)$-approximate solution for the problem is $\mathrm{NP}$-hard.
When surveying approximation algorithms for the problem, we distinguish between two
types of algorithms.

The first type uses a continuous relaxation
of the submodular set functions $f \colon 2^V \to \mathbb R$ to
$F \colon [0, 1]^V \to \mathbb R$, which is called the multilinear extension.
After obtaining a continuous solution via techniques from convex optimization,
the solution is then rounded to output a single set. Such approaches
have the benefit that they (almost) offer the best possible approximation ratio
of $1 - \frac 1 e$. However, evaluating the multilinear extension is costly
and impractical for even moderately sized problems, as we can generally only
approximate the relaxation via sampling.
In their seminal work, \citet{Chekuri+10} introduced the oblivious
rounding scheme \emph{Swap Rounding} and showed how this leads to an
algorithm for the multiobjective problem even subject to a matroid constraint.
Later, \citet{Udwani18} introduced a more efficient algorithm for
cardinality constraints. \citet{Tsang+19} gave an algorithm that
finds a continuous solution via Frank-Wolfe, and they showed that the
multilinear extension can be evaluated efficiently for influence maximization.

The second type of algorithm avoids evaluation of the multilinear extension
and only operates on discrete solutions. \citet{Udwani18} also introduced
an efficient $(1 - \frac 1 e)^2$-approximation algorithm via MWU. 
The algorithm maintains a vector of weights $y \in \Delta_C$, 
where $C$ is the index set of multiple submodular functions and $\Delta_C$ is the probability simplex over $C$, 
and uses a subroutine that solves the problem of maximizing the combined submodular function  
$\sum_{c \in C} y_c f_c(S)$. 
The weights are then updated based on the
individual values achieved for each objective. However, combining the solutions
created in each iteration introduces an additional factor of $1 - \frac 1 e$
into the approximation guarantee, which is the reason that this algorithm
does not retain the guarantee of continuous algorithms.
\citet{Orlin+18} devised a combinatorial algorithm that almost achieves the
$(1- \frac 1 e)$-approximation, but only for the case when $k\leq 2$.

Other efficient algorithms use the greedy paradigm, but do not obtain
constant-factor approximation guarantees.
The aforementioned bi-criteria algorithm by \citet{Krause+08} maximizes a modified objective function,
but may
violate the budget constraint in order to meet its approximation guarantee. 
Later \citet{Chen+17}, \citet{Wilder18}, and \citet{Anari+19} also gave bi-criteria approximation algorithms.

Instead of maximizing the minimum of multiple objectives, \citet{Soma+17} considered \emph{regret ratio minimization} where the goal is to find a collection of subsets (rather than a single subset) that approximate the Pareto optimal set to the multiple objectives. They designed algorithms for the problem which were improved by \citet{Feng+21} and \citet{Wang+23}. 
\citet{Malherbe+22} considered an objective based on quantile maximization. 
Recently, \citet{Wang+24} introduced the problem of maximizing the average of objectives subject to multiple constraints ensuring the value of each objective. 
\citet{Fazzone+24} studied a more general problem 
 and extended the algorithm of \citet{Krause+08} to obtain a sequence of approximately Pareto optimal solutions.

Related to the multiobjective problem, 
submodular maximization has also been studied with fairness and diversity considerations for the output subset. 
In this context, each element is associated with a categorical attribute such as race or gender, 
and the output is required to ensure that no particular attribute value is under- or over-represented.  
This problem has actively been studied in both the offline setting~\cite{Celis+18,ElHalabi+24,Tang+23,Yuan+23} 
and the streaming setting~\cite{ElHalabi+20,ElHalabi+23,Wang+21,Cui+24}.

\section{Problem Definition}\label{sec:problem}

We now formally define the notion of submodularity and the problem we
consider in this work.
A set function on a finite universe $V$ is a function $f \colon 2^V \to \R$
on all subsets of $V$. 
A set function $f$ is said to be \emph{monotone} if $f(S)\leq f(T)$ for all sets $S\subseteq T\subseteq V$.
For sets $S\subseteq V$ and $T\subseteq V$,
we define the \emph{marginal gain} of $T$ over $S$ as
$f(T\mid S)=f(S\cup T)-f(S)$.
If $T$ consists of a single element $v$, we simply write $f(v\mid S)$ instead
of $f(\{v\}\mid S)$.
A set function is called \emph{submodular} if it has diminishing returns,
i.e., for all sets $S \subseteq T \subseteq V$ and elements $v \in V$ holds that
\[
  f(v \mid S) \ge f(v \mid T) .
\]
Maximization of submodular functions is well-understood. 
We consider
(multiplicative) approximation ratios, so we assume throughout that submodular
functions are nonnegative.

In our work, we are
interested in a generalization of this objective where we consider multiple
submodular functions over the same universe. The goal is to find a set
that maximizes all functions at the same time:

\begin{problem}[Multiobjective Submodular Maximization]
  \label{prob:multiobjective}
  Let $C$ be a finite set, which we refer to as the set of colors, and define $k = |C|$.
  We are given a universe $V$ and monotone and submodular
  set functions $f_c \colon 2^V \to \mathbb R_{\geq 0}$ for each $c \in C$, all on the
  same universe $V$.     Given a budget $B$, 
  we are asked to find a set $S \subseteq V$ of size $|S| \le B$
  maximizing the least value among all colors, which is $\min_{c \in C} f_c(S)$.
\end{problem}

In our work, we are concerned with the case when the budget is large compared
to the number of colors, i.e., $B \ge k$, as this avoids the regime where
the problem becomes inapproximable \cite{Krause+08}.
Submodular maximization (i.e., the special case of a single color) is known to
be NP-hard as it generalizes the maximum coverage problem.
For the latter, no polynomial-time algorithm with an approximation ratio
better than $1 - \frac 1 e$ is known. At the same time, the greedy algorithm
provably achieves this approximation ratio for any submodular
function~\cite{Nemhauser+78}. 
Multiobjective submodular maximization (i.e., the general case of multiple
colors) is a harder problem, but we aim for (almost) the same
approximation guarantee. As in the prior work, our goal is to design an
algorithm that provably obtains an approximation ratio of
$1 - \frac 1 e - \epsilon$ with probability $1 - \delta$ for any pair of
$\epsilon, \delta > 0$. Both values show up in the running time of our
algorithm and in a necessary condition on the budget $B$.

\paragraph{Further Notation.}

We use $\OPT$ for the optimal solution
$\OPT = \arg\max_{S \subseteq V : |S| \le B} \min_{c \in C} f_c(S)$
of the multiobjective submodular maximization 
(Problem~\ref{prob:multiobjective}).
Abusing notation, we also use $\OPT$ to denote the 
optimal value, i.e., we write $\OPT = \min_{c \in C} f_c(\OPT)$.
Given a finite set $X$, we denote the set of probability distributions
over $X$ as the simplex
$\Delta_X = \{ x \in [0,1]^X : \sum_{v \in X} x_v = 1 \}$.
We will use this to denote probability distributions over the set of colors
$\Delta_C$ or the universe $\Delta_V$.

\section{Our Algorithm}\label{sec:algorithm}

\begin{algorithm2e}[t]
\KwIn{Monotone and submodular set functions $f_c \colon 2^V \to \R_{\geq 0}$ for $c \in C$, budget $B$, and the optimal value $\OPT$}
$S^{(0)} \gets \emptyset$\;
\For{$i = 1, 2, \dots, B$}{
  Let $x^{(i)} \in \Delta_V$ be a solution to the problem:\\ Find $x \in \Delta_V$
  such that for all $c \in C$ holds
  \begin{align}
    \label{eq:4}
    \sum_{v \in V} x_v f_c(v \mid S^{(i-1)}) \ge \frac 1 B \left(\OPT - f_c(S^{(i-1)})\right);
  \end{align}\\
  Sample $v^{(i)} \sim x^{(i)}$\;
  Update $S^{(i)} \gets S^{(i-1)} \cup \{v^{(i)}\}$\;
}
\Return{$S^{(B)}$}.
\caption{Greedy for Problem~\ref{prob:multiobjective}}
\label{alg:multiobjective-simple}
\end{algorithm2e}

We sketch the central component of our algorithm for multiobjective submodular
maximization (Problem~\ref{prob:multiobjective}) in
Algorithm~\ref{alg:multiobjective-simple}. Note that this is merely to
illustrate and motivate our core ideas, and we also provide an implementable
version as Algorithm~\ref{alg:multiobjective}.
Before delving into the analysis of our approximation guarantee, we
want to motivate our algorithm design.
For now, we assume that we know the optimal value $\OPT$;
an assumption that we will later remove.

As in the greedy paradigm, we construct a solution $S^{(i)}$ over $1 \le i \le B$
iterations and in each iteration, we identify a good element $v^{(i)}$ to add.
For the case of submodular maximization (i.e., the special case of a single
color), using the element with maximum marginal gain
$v^* = \arg\max_{v \in V} f(v \mid S^{(i-1)})$ is sufficient to obtain the
$(1 - \frac 1 e)$-approximation.
However, the analysis reveals that it is not necessarily required to add the element
of maximum marginal gain. Indeed, adding any element $v \in V$ whose marginal
gain satisfies the inequality
\begin{align}
  \label{eq:17}
  f(v \mid S^{(i-1)}) \ge \frac 1 B \left(\OPT - f(S^{(i-1)})\right)
\end{align}
is sufficient. Intuitively, adding an element with a marginal gain that is
high enough compared to the distance to optimality is sufficient, and it is
easy to show that $v^*$ also satisfies Inequality~\eqref{eq:17}.
Transferring this to the case of multiple colors is not straightforward as there
is not necessarily a single element that satisfies Inequality~\eqref{eq:17}
simultaneously for all colors $c \in C$. Instead, we show that we can satisfy
\eqref{eq:17} in expectation with a probability distribution
$x^{(i)} \in \Delta_V$ which is exactly what we require for \eqref{eq:4} in the
description of Algorithm~\ref{alg:multiobjective-simple}.
As of now, it is unclear that such a probability distribution exists, but
show this in the following lemma:

\begin{lemma}
  \label{lem:opt-is-good}
  In each iteration of Algorithm~\ref{alg:multiobjective-simple}, there is
  a solution $x^{(i)} \in \Delta_V$ that satisfies Inequality~\eqref{eq:4}.
\end{lemma}

\begin{proof}
  Fix an iteration $1 \le i \le B$. 
  Let $x^* \in [0,1]^V$ be such that $x^*_v = \frac 1 B$ if $v \in \OPT$
  and $x^*_v = 0$ otherwise. Note that $x^*$
  satisfies Inequality~\eqref{eq:4} for all $c \in C$ since
  \begin{align*}
    &\sum_{v \in V} x^*_v f_c(v \mid S^{(i-1)}) \\
    &= \frac 1 B \sum_{v \in \OPT} f_c(v \mid S^{(i-1)}) \\
    &\ge \frac 1 B f_c(\OPT \mid S^{(i-1)}) & \textrm{(submodularity)} \\
    &\ge \frac 1 B \left( f_c(\OPT) - f_c(S^{(i-1)}) \right) & \textrm{(monotonicity)} \\
    &\ge \frac 1 B \left( \OPT - f_c(S^{(i-1)}) \right) . & \hfill \qedhere
  \end{align*}
\end{proof}

Even though we know that a solution $x^{(i)}$ exists in each iteration,
we still need to argue that we can find it efficiently.
To this end, we rewrite the problem of finding $x^{(i)}$ satisfying Inequality~\eqref{eq:4} as an LP. 
Introducing an objective into the LP even allows us to remove the dependence on $\OPT$:
\begin{align*}
  \textrm{LP$(S^{(i-1)})$: } & \text{ max. } \xi \numberthis \label{eq:6} \\
  \textrm{ s.t. }
  \forall c \in C \colon & \! \sum_{v \in V} x_v\! \left(
  B f_c(v \mid S^{(i-1)})\! +\! f_c(S^{(i-1)})
  \right) \ge \xi, \\
  & \! \sum_{v \in V} x_v = 1, \\
  \forall v \in V \colon & x_v \ge 0. 
\end{align*}
Note that the equivalence with \eqref{eq:4} follows immediately by rearranging
terms in the first constraint using the second constraint
$\sum_{v \in V} x_v = 1$ and replacing $\xi$ with $\OPT$.
It is thus guaranteed by Lemma~\ref{lem:opt-is-good} that the optimal value
$\xi^*$ of the LP satisfies $\xi^* \ge \OPT$.

Now, operating over continuous solutions $x^{(i)}$ is impractical as it involves relaxing
the submodular functions. On the contrary, we are able to avoid continuous
solutions as we discretize immediately by sampling a
random element $v^{(i)} \in V$ according to $x^{(i)}$ and adding it to our
solution. The technical part of our analysis is concerned with showing that this
leads to a solution of high value for all colors $c \in C$, with sufficiently
high probability.
However, intuitively this is clear: Assume we have constructed a partial
solution $S^{(i)}$ that has low value for a specific color $c \in C$.
This means that the gap to optimality $\OPT - f_c(S^{(i-1)})$ is larger for
$c$ than for other colors, so the solution $x^{(i-1)} \in \Delta_V$ will
put more mass on elements $v \in V$ that satisfy \eqref{eq:4} for the color
$c$. This makes it more likely that $f_c(S^{(i)})$ will be larger, and overall
that all colors have high value, given that the budget is sufficiently large.

\begin{algorithm2e}[t]
\KwIn{Monotone and submodular set functions $f_c \colon 2^V \to \R_{\geq 0}$ for $c \in C$, budget $B$, failure probability $\delta$}
$S_{\mathrm{best}} \gets \emptyset$\;
\For{$t = 1, 2 ,\dots, \lceil \log(2 / \delta) \rceil$}{
  $S^{(0)} \gets \emptyset$\;
  \For{$i = 1, 2, \dots, B$}{
    Let $x^{(i)} \in \Delta_V$ be a solution to LP$(S^{(i-1)})$ \;
    Sample $v^{(i)} \sim x^{(i)}$\;
    Update $S^{(i)} \gets S^{(i-1)} \cup \{v^{(i)}\}$ \;
  }
  \If{$\min_{c \in C} f_c(S^{(B)}) \ge \min_{c \in C} f_c(S_{\mathrm{best}})$}{
    $S_{\mathrm{best}} \gets S^{(B)}$ \;
  }
}
\Return{$S_{\mathrm{best}}$}.
\caption{LP Greedy with Independent Repetitions}
\label{alg:multiobjective}
\end{algorithm2e}

Finally, we need to amplify the success probability of our algorithm, which
we do via independent repetitions of Algorithm~\ref{alg:multiobjective-simple}.
We detail our complete approach in Algorithm~\ref{alg:multiobjective}, which
also no longer requires knowledge of $\OPT$.

Adding elements randomly to the solution is related to other prior works
on submodular maximization. For instance, the algorithm of
\citet{Buchbinder+14} addresses non-monotone submodular maximization by using
randomness to avoid adding bad elements to the solution.
A difference to our approach is that our sampling probabilities are carefully
chosen to solve the multiobjective problem. This becomes obvious when running
our algorithm on a single color, in which case our algorithm reduces to the
standard greedy algorithm. This shows that we choose our sampling probabilities
to enforce progress across all colors simultaneously.

\subsection{Analysis Outline}
\label{sec:analysis}

We now state the approximation guarantee of
Algorithm~\ref{alg:multiobjective-simple}. This guarantee only holds if the
optimal value $\OPT$ is sufficiently large, but we will show in
Section~\ref{sec:removing-condition} how to convert this
into a guarantee that depends on the ratio of the budget $B$ and the number of
colors $k = |C|$ via a pre-processing
step. To bound the running time and the number of function evaluations of our algorithm, we
will show in Section~\ref{sec:lazy} how to use MWU to solve the LP efficiently.
We defer some of our proofs to Appendix~\ref{sec:appendix-omitted-proofs}.

\begin{restatable}{theorem}{multiobjective}
  \label{thm:multiobjective}
  If $\OPT \ge \frac{4}{\epsilon^2} M \log(2 k)$ where the maximum marginal
  gain is
  $M = \max_{c \in C} \max_{v \in V} f_c(v \mid \emptyset)$,
    then Algorithm~\ref{alg:multiobjective} outputs a solution $S$ such that
  \[
    \min_{c \in C} f_c(S) \ge \left( 1 - \frac 1 e - \epsilon \right) \OPT
  \]
  with probability at least $1 - \delta$.
  \end{restatable}

To prove the theorem, we begin with the basic observation that in each iteration, we add an element to
the solution that for each $c \in C$ is good in expectation.

\begin{restatable}{lemma}{expectationisgood}
  \label{lem:expectation-is-good}
  In each iteration $i$ of Algorithm~\ref{alg:multiobjective-simple} and for
  each color $c \in C$,
  \[
    \mathbb E[ f_c(v^{(i)} \mid S^{(i-1)}) \mid S^{(i-1)}] \ge
    \frac 1 B \left( \OPT - f_c(S^{(i-1)}) \right) .
  \]
\end{restatable}

In expectation, we recover the greedy guarantee:

\begin{restatable}{lemma}{expectedvalue}
  \label{lem:expected-value}
  Algorithm~\ref{alg:multiobjective-simple} outputs a solution $S$ such that for each $c \in C$,
  \[
    \E[ f_c(S) ] \ge \left( 1 - \frac 1 e \right) \OPT .
  \]
\end{restatable}

The problem is that this is only in expectation, but we have to satisfy it
for all colors $c \in C$ simultaneously. To this end, we show how to use
concentration results for martingales.

To motivate our analysis further we explain how this generalizes an
easier problem that is better understood: Imagine that in each iteration $i$, we
obtain $x^{(i)} = x^*$ as defined in the proof of Lemma~\ref{lem:opt-is-good} as a solution
to \eqref{eq:4}. Since this solution does not change over the iterations, our
algorithm is equivalent to throwing $B$ balls (one ball per iteration) into $B$
bins (one bin per element of $\OPT$). It is known that the resulting number of
non-empty bins is $1 - \frac 1 e$ in expectation. In fact, we have shown
a stronger statement in Lemma~\ref{lem:expected-value}, i.e., that this even holds for
functions that are submodular over the non-empty bins. Our analysis thus takes
inspiration from concentration results for this balls-into-bins problem to obtain
concentration.
In particular, we use a form of Azuma's inequality due to \citet{Kuszmaul+21}:

\begin{theorem}[Corollary 16 in \citet{Kuszmaul+21}]
  \label{thm:azumas}
  Suppose that Alice constructs a sequence of random variables $X_1, \dots, X_B$,
  with $X_i \in [0, M]$, $M > 0$, using the following iterative process.
  Once the outcomes of $X_1, \dots, X_{i-1}$ are determined, Alice then selects
  the probability distribution $\mathcal D_i$ from which $X_i$ will be drawn;
  $X_i$ is then drawn from distribution $\mathcal D_i$. Alice is an adaptive adversary
  in that she can adapt $\mathcal D_i$ to the outcomes of $X_1, \dots, X_{i-1}$.
  The only constraint on Alice is that $\sum_{i=1}^B \E[X_i \mid \mathcal D_i] \ge \mu$,
  that is, the sum of the means of the probability distributions
  $\mathcal D_1, \dots, \mathcal D_B$ must be at least $\mu$. 
    If $X = \sum_{i=1}^B X_i$, then for any $\epsilon > 0$,
  \[
    \Pr[X \le (1 - \epsilon) \mu] \le \exp\left( - \frac{\epsilon^2 \mu}{2 M} \right) .
  \]
\end{theorem}

We use this to argue for concentration around (the lower bound on) the
expectation, which was $(1 - \frac 1 e) \OPT$:

\begin{restatable}{lemma}{individualconcentration}
  \label{lem:individual-concentration}
  Let $S$ be the output of Algorithm~\ref{alg:multiobjective-simple}. For each $c \in C$, it  holds that 
  \[
    \Pr\left[f_c(S) \le \left(1 - \frac 1 e - \epsilon \right) \OPT \right]
    \le \exp\left( - \frac{\epsilon^2 \OPT}{4 M_c} \right)
  \]
  where $M_c = \max_{v \in V} f_c(v \mid \emptyset)$. \end{restatable}

\begin{proof}
  Fix $c \in C$.
  In the context of Theorem~\ref{thm:azumas}, we use
  $X_i = f_c(v^{(i)} \mid S^{(i-1)})$
  and the distribution $\mathcal D_i$ over values
  $f_c(v^{(i)} \mid S^{(i-1)})$ where $v^{(i)} \sim x^{(i)}$.
  As such, $M = \max_{c \in C} M_c$ is also the maximum marginal value.
  Note that the
  random process is exactly as in Algorithm~\ref{alg:multiobjective-simple}:
  In each iteration, we select a (possibly adversarial) solution $x^{(i)}$
  to~\eqref{eq:4} which determines the distribution $\mathcal D_i$.
  It remains to argue that $\sum_{i=1}^B \E[X_i \mid \mathcal D_i] \ge \mu$ 
  for $\mu = (1 - \frac 1 e) \OPT$. 
    However, we have just shown this in Lemma~\ref{lem:expected-value}:
  \[
    \sum_{i=1}^B \E[X_i \mid \mathcal D_i] =
    \E\left[ \sum_{i=1}^B f_c(v^{(i)} \mid S^{(i-1)}) \right] =
    \E[ f_c(S^{(B)}) ] \ge \left( 1 - \frac 1 e \right) \OPT = \mu.
  \]
  Applying Theorem~\ref{thm:azumas} now yields
  \[
    \Pr[f_c(S) \le (1 - \epsilon) \mu] \le
    \exp\left( - \frac{\epsilon^2 \mu}{2 M_c} \right) .
  \]
  Noticing that $\left(1 - \frac 1 e - \epsilon \right) \OPT \le (1 - \epsilon) \mu$ and $1-\frac{1}{e}> \frac{1}{2}$, 
  we have the lemma. 
    \end{proof}

With the concentration for each color in hand,
we are now ready to prove Theorem~\ref{thm:multiobjective}:

\begin{proof}[Proof of Theorem~\ref{thm:multiobjective}]
Recall that we defined the maximum gain as $M = \max_{c \in C} M_c
= \max_{c \in C} \max_{v \in V} f_c(v \mid \emptyset)$. Assuming $\OPT \ge \frac 4 {\epsilon^2} M \log (2 k)$, we have by
Lemma~\ref{lem:individual-concentration} for each $c \in C$ that
\[
  \Pr\left[f_c(S) \le \left(1 - \frac 1 e - \epsilon\right) \OPT\right]
  \le \exp\left( - \frac{\epsilon^2 \OPT}{4 M} \right)
  \le \frac 1 {2 k}
\]
and by a union bound over the $k$ (dependent) colors, we obtain that
the result holds for all $c \in C$ with probability at least $\frac 1 2$.
We repeat the algorithm $r = \lceil \log(1 / \delta) \rceil$ times and output
the best solution.  Since the repetitions are independent, the probability that
none of the $\lceil \log(1 / \delta)\rceil$ solutions is a
$(1 - \frac 1 e - \epsilon)$-approximation is
$2^{-r} \le e^{-\log(1 / \delta)} = \delta$ as promised.
\end{proof}

\subsection{Removing the Condition on $\OPT$}
\label{sec:removing-condition}

We are able to push the condition on $\OPT$ into a condition on the relation of
the budget $B$ to the number of
colors $k = |C|$, which we state in Theorem~\ref{thm:final}.
As do all $(1 - \frac 1 e - \epsilon)$-approximation algorithms for
Problem~\ref{prob:multiobjective},
we also require the budget to be large. Recall that this
avoids the regime where the problem becomes inapproximable~\cite{Krause+08}
and is natural for many of the applications we mentioned in the introduction.
Our condition is slightly less restrictive compared to the prior work, and we
provide a detailed comparison at the end of this section.

The key idea is that we use a pre-processing step where we eliminate elements
$v \in V$ whose value $f_c(v \mid \emptyset)$ is large compared to $\OPT$ for
any color $c \in C$, as such elements impair the concentration.
We find such elements and add them to a partial solution $T$. After the
pre-processing, some colors may have reached $\OPT$. However,
we only worry about the set
\[
  \tilde C = \{ c \in C : f_c(T) < \OPT \}
\]
of colors which have not yet reached $\OPT$.
Our pre-processing guarantees that marginal gains for colors $c \in \tilde C$
on top of the partial solution $T$
are small compared to $\OPT$, and are thus not problematic for
Algorithm~\ref{alg:multiobjective}.
This is similar to the pre-processing step described by \citet{Udwani18}, but
our pre-processing does not require knowledge of $\OPT$.
We defer the full description of the pre-processing routine to
Algorithm~\ref{alg:pre-processing} in Appendix~\ref{sec:appendix-removing-condition}, along with all proofs.

We then run Algorithm~\ref{alg:multiobjective} with a reduced
budget $B - |T|$ on functions $\tilde f_c(A) = f_c(A \cup T)$, which
ensures by submodularity that marginal gains are small for the modified
instance. The technical difficulty is to show that reducing the budget does
not result in a big decrease of the optimum, which is non-trivial for
multiple objectives.
We provide the following novel result:

\begin{restatable}{lemma}{optchange}
  \label{lem:opt-change}
  Let $\widetilde \OPT_b = \max_{S \subseteq V : |S| \le b} \min_{c \in \tilde C} \tilde f_c(S)$
  for $b\ge 0$.
    Let $\gamma > 0$ be such that
  $\tilde f_c(v \mid \emptyset) \le \gamma \widetilde \OPT_B$
  for all $v \in V$ and $c \in \tilde C$.
  Then, for $\tilde B \le B$,
    \[
    \widetilde \OPT_{\tilde B} \ge \left(1 -
    \sqrt{3 \gamma \frac B {\tilde B} \log k} \right) \frac {\tilde B} B
    \widetilde \OPT_B .
  \]
\end{restatable}

Note that prior work uses similar results but over the continuous space,
which makes our result more challenging.
Let $\tilde S$ be the solution of Algorithm~\ref{alg:multiobjective} to the
modified instance. We have the following guarantee on the combined
solution:

\begin{restatable}{theorem}{final}
  \label{thm:final}
  If $B \ge 108 \frac k {\epsilon^3} \log k$ then
  \[
    \min_{c \in C} f_c(T \cup \tilde S) \ge
      \left(1 - \frac 1 e - \epsilon \right) \OPT
  \]
  with probability at least $1 - \delta$.
\end{restatable}

\paragraph{Comparison with prior work.}

We compare our condition $B = \Omega(\frac k {\epsilon^3} \log k)$ of
Theorem~\ref{thm:final} with conditions on the budget of prior work.
\citet{Udwani18} requires that $\epsilon \le \frac 1 {10 \log k}$
for an approximation ratio of
$1 - \frac 1 e - \epsilon - \frac k {B \epsilon^3} - (\frac{\log B} B)^{1/2}$.
In order to achieve an approximation ratio of
$1 - \frac 1 e - \epsilon$, his algorithm requires that
$B = \Omega( \frac k {\epsilon^4} )$.
The algorithm of \citet{Tsang+19} provides a similar approximation
that is also dominated by the term $\frac{k}{B \epsilon^3}$. Since they have to use 
the same choice of $\epsilon$, their algorithm has the same condition on $B$.
The condition on $\epsilon$ in both works
requires $\frac 1 \epsilon \ge 10 \log k$, which means that
our condition for Theorem~\ref{thm:final} is never more restrictive, while
our condition is strictly less restrictive in its dependence on $\epsilon$.

\subsection{MWU and Lazy Evaluations}
\label{sec:lazy}

Algorithm~\ref{alg:multiobjective} may still be impractical for two
reasons. First, solving a general LP may be expensive and second,
we need to evaluate all marginal gains in each iteration. The latter is also
a problem appearing in the normal greedy algorithm (for a single color) and we
there resort to lazy evaluations.
It turns out that we can solve both efficiency problems via multiplicative
weights updates (MWU)~\cite{Arora+12}.
In particular, we show that solving the LP approximately
via a few rounds of MWU is sufficient, and we even transfer
lazy evaluations to multiple colors.

We detail this in Algorithm~\ref{alg:mwu} in Appendix~\ref{sec:appendix-lazy}.
We crucially use
that solving $\mathrm{LP}(S)$ admits a formulation as a two-player game:
The first player selects a probability distribution $y \in \Delta_C$ over the
colors, and the second player responds with the element $v \in V$ which
maximizes the convex combination $\sum_{c \in C} y_c \ell_c(v)$ where
$\ell_c(v)$ is a loss (or rather, a gain) given as
$\ell_c(v) = B f_c(v \mid S) + f_c(S)$. Finding the best response thus entails
a search over all elements $v \in V$, and we show how to transfer the concept of
lazy evaluations to the multiobjective problem in order to reduce the number
of function evaluations in practice.
Furthermore, two-player games can be solved efficiently using few rounds
of MWUs~\cite{Arora+12}. The complication is to bound the number of rounds, and
we do so by appealing to the structure of our problem.
This allows us to bound the overall running time of our algorithm.
Note that we first
run the pre-processing described in Algorithm~\ref{alg:pre-processing}, and
then apply 
Algorithm~\ref{alg:multiobjective}
where we solve the LP
approximately using MWU as described in Algorithm~\ref{alg:mwu}.

\begin{restatable}{theorem}{runningtime}
  \label{lem:running-time}
  Our algorithm for multiobjective submodular maximization runs in time
  $O(n B^3 \frac 1 {\epsilon^2} k \log (k) \log(1 / \delta))$ and requires
  $O(n B k \log (1 / \delta))$ function evaluations.
\end{restatable}

We also defer this proof to Appendix~\ref{sec:appendix-lazy}.
We compare this with the $(1 - \frac 1 e - \epsilon)$-approximation algorithm of
\citet{Udwani18}, which has a similar condition on $B$ but needs
$\tilde O(k n^8)$ function evaluations.
The faster $((1 - \frac 1 e)^2 - \epsilon)$-approximation algorithm requires
$\tilde O(n \frac 1 {\epsilon^3})$ function evaluations by repeatedly solving
submodular maximization problems using a nearly-linear time algorithm
developed by \citet{Mirzasoleiman+15}.
Furthermore, the $(1- \frac 1 e - \epsilon)$-approximation algorithm of \citet{Tsang+19} needs
$\tilde O(\frac{k B^2}\epsilon + \frac{B^4}{\epsilon^5})$ evaluations of the multilinear extension and
its gradient, and needs $O(\frac{n B^2}{\epsilon^2} + \frac{kB^2}{\epsilon} + \frac{B^3}{\epsilon^2})$ additional time.
Note that the evaluation of the multilinear extension is generally not tractable
and usually estimated via sampling~\cite{Salem+23}.
Only special structure of the submodular functions allows for more efficient
evaluations.

\section{Fair Centrality Maximization}
\label{sec:fcm}

We now introduce our novel application
which furthers the applicability of
algorithms for multiobjective submodular maximization.
Centrality measures, quantifying the importance of nodes or edges in a network, play a key
role in network analysis~\cite{Das+18}. In many real-world scenarios, we want to optimize 
the centrality score of a target node by intervening in the
structure of the network (e.g., by adding or removing edges around the node). 
In particular, \emph{centrality maximization} is a well-studied task, 
where given a directed graph $G=(V,A)$, a target node
$v\in V$, and a budget $B\in \mathbb{Z}_{>0}$, we are asked to insert at most
$B$ edges heading to $v$ that maximize the centrality score of $v$ (see e.g., \cite{Bergamini+18,Crescenzi+16,DAngelo+19,Ishakian+12,Medya+18}).

Among existing centrality measures, the \emph{harmonic centrality} is one of the most well-established~\cite{Boldi+14}. 
The harmonic centrality score of $v\in V$ is defined as 
\begin{align*}
  h_G(v)=\sum_{u\in V\setminus \{v\}}\frac{1}{d_G(u,v)}
\end{align*}
where $d_G(u,v)$ is the shortest-path distance from $u$ to $v$ on $G$. 
Intuitively, the score quantifies the importance of nodes based on the
level of reachability from the other nodes. 
\citet{Boldi+14} showed
that unlike previously known centrality measures, the harmonic centrality
satisfies all the desirable axioms, namely the size axiom, density axiom, and
score monotonicity axiom. Recently, \citet{Murai+19}
theoretically and empirically demonstrated that among well-known centrality
measures, the harmonic centrality is most stable and thus reliable against the
structural uncertainty of networks.

In many real-world networks, nodes are not uniform but varied in terms of
attributes. Examples include social networks, where nodes have sensitive
attributes such as race, gender, religion, or even political opinions. Suppose
that we have a network in which each node has a categorical attribute such as
race or gender. Centrality maximization without any consideration of the
variation of node attributes may lead to undesirable outcomes. 
Indeed, as the objective function, i.e., the centrality measure employed, does not take
into account the variation of node attributes, it cannot distinguish between
parts of centrality scores of the target node corresponding to nodes with
different attribute values. Therefore, even if the centrality score of the
target node is maximized, the node might still not be sufficiently visible to
nodes with some specific attribute value. This is problematic, for instance, in the case where
we use centrality maximization to improve the visibility of public health
agency's accounts in social media platform. Such accounts should be sufficiently
visible even to minority users.

We therefore study centrality maximization with fairness considerations, taking into
account the variation of node attributes.
To this end, we introduce a novel centrality measure, which we refer to as the
\emph{fair harmonic centrality}. This measure is a generalization of the
aforementioned harmonic centrality, 
which 
contributes to finding a fair solution in terms of node attributes.
Let now $C$ be a set of colors, i.e., attribute values. 
Let $\ell \colon V\rightarrow C$ be a mapping that assigns each node to a color. 
For each $c\in C$, define $V_c=\{v\in V\mid \ell(v)=c\}$.
We define the fair harmonic centrality based on the maximin
fairness~\cite{Rahmattalabi+19} as follows:
\begin{align}
\label{eq:18}
  h^\mathrm{min}_G(v)=\min_{c\in C} \frac{1}{|V_c\setminus \{v\}|} \sum_{u\in V_c\setminus \{v\}}\frac{1}{d_G(u,v)}. 
\end{align}
This represents the minimum value among all parts of the harmonic centrality score of $v$ 
corresponding to nodes with different colors, normalized by their populations.
Clearly, the above is a generalization of the original harmonic centrality.
Based on this measure, we formulate:

\begin{problem}[Fair Centrality Maximization]
\label{prob:fcm}
Given a directed graph $G=(V,A)$, a mapping $\ell \colon V\rightarrow C$, a target node $v\in V$, and a budget $B\in \mathbb{Z}_{>0}$, we are asked to find 
$F\subseteq \{u\in V\mid (u,v)\notin A\}\times \{v\}$ 
with $|F|\leq B$ that maximizes 
$h^\mathrm{min}_{(V,A\cup F)}(v)$. 
\end{problem}

This is a special case of the multiobjective submodular maximization (Problem~\ref{prob:multiobjective}), 
since every term in the minimum in \eqref{eq:18} is submodular
in $F$, as in the original harmonic centrality~\cite{Crescenzi+16}.
Existing algorithms for multiobjective submodular maximization either
offer weak theoretical guarantees
or rely on the
multilinear extension of the objective function; however, the evaluating the latter
is computationally expensive when considering our fair centrality maximization problem
for large-scale networks. With our algorithm, we are the first to solve
multiobjective submodular maximization on this scale with
the strongest~guarantee.

\section{Experimental Evaluation}\label{sec:experiments}

We evaluate our algorithm, natural greedy baselines, and other
prior work on
synthetic and real-world instances for
max-$k$-cover, fair centrality maximization, and fair influence maximization.
We run our experiments in Python 3 on a ThinkPad X1 Carbon with an Intel Core i7-1165G7 CPU
and 16GB of RAM. Our code is publicly
available.\footnote{\url{https://github.com/285714/multiobjective}}

\subsection{Algorithms}

We use our algorithm and several baseline algorithms with and without theoretical guarantees.
For baselines that require an estimate $\OPT'$ to the optimal value,
we use an outer loop that performs a binary search for $\OPT'$.
Note that our algorithm does not require such an estimate.

\paragraph{LP Greedy}
We use our Algorithm~\ref{alg:multiobjective} with 20 repetitions of
the outer loop to boost the success probability, and solve
LP~\eqref{eq:6} using Gurobi Optimizer 11.0.1~\cite{gurobi}. We facilitate lazy evaluations
differently from the
MWU described in Algorithm~\ref{alg:mwu}:
In each iteration $i$, we solve $\mathrm{LP}(S^{(i-1)})$, but using the upper
bounds $g_c(v)$ in place of the real marginal gains. Whenever the solution
$x \in \Delta_V$ to the LP places mass $x_v > 0$ on an element $v \in V$, we
evaluate and update $g_c(v) = f_c(v \mid S^{(i-1)})$. We then solve the LP again
and repeat this until the real marginal gains of all elements $v \in V$ with $x_v > 0$
are evaluated. We omit the pre-processing described in
Algorithm~\ref{alg:pre-processing}.
Algorithm~\ref{alg:multiobjective} may act overly conservative in its
selection to preserve the approximation ratio for difficult instances.
We thus modify the left-hand-side of the first constraint of LP \eqref{eq:6} to be 
\[
  \sum_{v \in V} x_v \left( B f_c(v \mid S^{(i-1)}) + \phi f_c(S^{(i-1)}) \right)
\]
where $\phi > 1$ is a factor that controls the greediness of our algorithm:
For larger $\phi$, our algorithm prefers to increase the color that currently
has the least value, instead of picking a distribution $x$ that leads to a
balanced increase.
Throughout, we use a factor $\phi=10$.
We include an ablation study for the
number of repetitions of the outer loop and the factor~$\phi$.

\paragraph{Greedy heuristics}
We use two greedy heuristics. First, we use the
heuristic described by \citet{Udwani18} (\textsc{Greedy Round Robin}):
In $1 \le i \le B$ iterations, we add one element
after another to the solution $S^{(i-1)}$. We select the $i$-th element as
$v^{(i)} = \arg\max_{v \in V} f_{c^*}(v \mid S^{(i - 1)})$ where
$c^* = i \mod k$. We use another greedy heuristic (\textsc{Greedy Minimum}),
where $c^* = \arg\min_{c\in C} f_c(S^{(i-1)})$.

\paragraph{Saturate}
We use the bi-criteria algorithm \textsc{Saturate} due to~\citet{Krause+08}
as a heuristic with fixed budget.
The algorithm uses a guess $\OPT'$ to the optimal value and
optimizes greedily over the submodular function
$f_{\textsc{Saturate}}(S) = \sum_{c \in C} \min \{f_c(S) , \OPT'\}$.

\paragraph{Udwani's MWU}
We use the efficient $(1 - \frac 1 e)^2$-approximation described by \citet{Udwani18}
(\textsc{Udwani MWU}).
We follow his implementation details and also omit the pre-processing. We use 100 iterations with a step size of $\eta=0.1 \cdot \OPT'$.
We use the lazy greedy implementation to solve the inner maximization problem.

\paragraph{Continuous Frank-Wolfe for influence maximization}
We use the algorithm of \citet{Tsang+19} and their implementation. The
algorithm relies on the fast evaluation of the multilinear extension
for fair influence maximization, and thus we use it only for fair influence
maximization.

\subsection{Max-$k$-Cover}

\begin{figure}
\centering
\includegraphics[width=.5\linewidth]{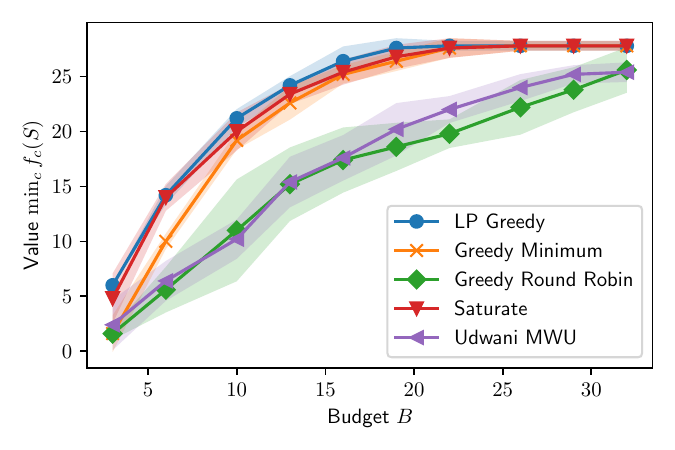}~
\includegraphics[width=.5\linewidth]{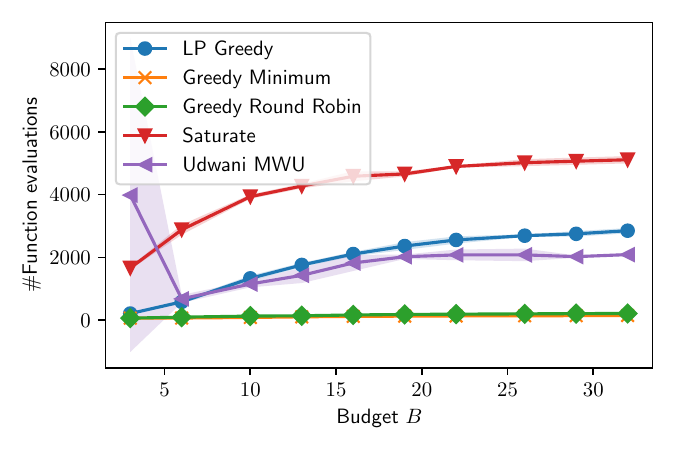}
\caption{
  Multiobjective submodular maximization for max-$k$-cover. We use $k=20$
  Kronecker graphs on $n=64$ nodes. We show the function value (left) and
  the number of evaluations (right). We report mean and standard deviation
  over 5 random instances.
  }
\label{fig:coverage-kronecker}
\end{figure}

First, we replicate the setup of \citet{Udwani18} for max-$k$-cover instances:
Here, we generate $k$ random synthetic graphs $\{G_c=(V,E_c)\}_{c \in C}$ on a fixed
vertex set $V$. The cover size of $U\subseteq V$ on graph $G_c$ is $f_c(U) = |N_{G_c}(U)|$ 
where $N_{G_c}(U)=\{\{u,v\}\in E_c : u\in U \text{ or } v\in U\}$.
We use $k=20$ stochastic Kronecker graphs on $n=64$ nodes which reflect
real-world networks and are detailed in \citet{Udwani18} and \citet{Leskovec+10}. Our
results are in Figure~\ref{fig:coverage-kronecker} and
we include further experiments on other graphs in Appendix~\ref{sec:appendix-max-coverage}.
This shows that we achieve the highest objective with
fewer function evaluations than other algorithms, particularly compared with
\textsc{Udwani MWU} which is the only other algorithm with theoretical guarantees.

\subsection{Fair Centrality Maximization}

\begin{figure}
\centering
  \includegraphics[width=0.5\linewidth]{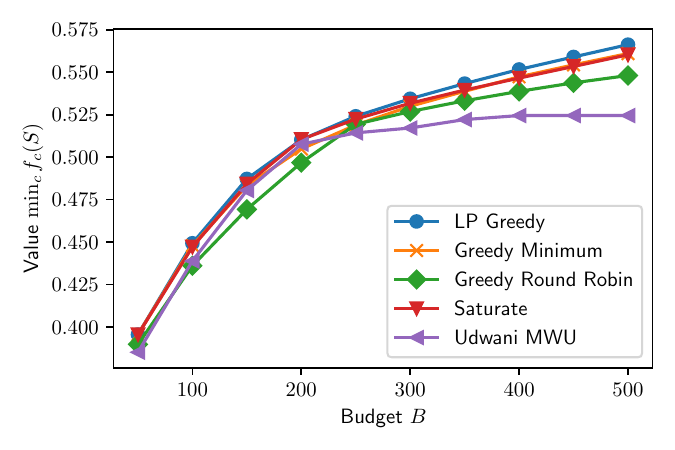}~
  \includegraphics[width=0.5\linewidth]{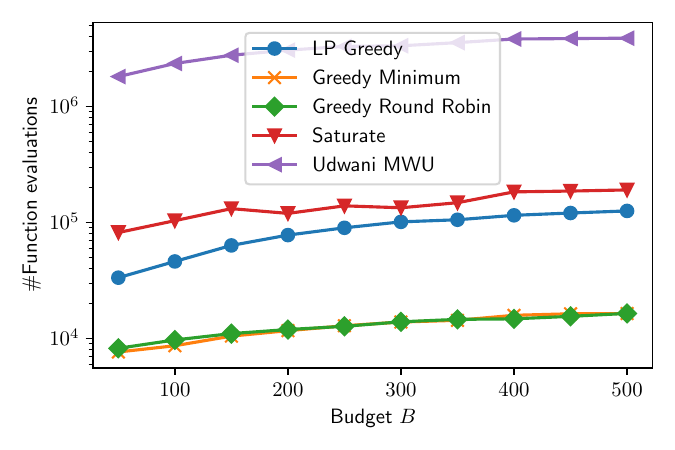}
  \caption{
  Fair centrality maximization on the Amazon co-purchasing graph
  \emph{Arts, Crafts \& Sewing} with $n=5051$ nodes and $k=2$ colors.
}
\label{fig:fcm-amazon1}
\end{figure}

We use Amazon co-purchasing networks used by \citet{Anagnostopoulos+20} and \citet{Miyauchi+23} along with
the color attributes which represent product categories to obtain instances
for our fair centrality maximization problem as defined in Section~\ref{sec:fcm}.
The networks are available online \cite{amazondatasets} and contain
large graphs of up to $10\ 380$ nodes and $53\ 680$ edges. 
We select an arbitrary target node among the nodes with median degree.
Figure~\ref{fig:fcm-amazon1} shows results for a single network called \emph{Arts, Crafts \& Sewing} 
and we include results for other networks in Appendix~\ref{sec:appendix-fair-centrality-max}.
Throughout, our algorithm achieves the highest objective value with fewer function evaluations than
\textsc{Udwani MWU} and \textsc{Saturate}, which are the strongest competitors.
This is mainly because, compared to these two algorithms, our method does not
require a binary search to determine a good guess to the optimal value.

\subsection{Fair Influence Maximization}

\begin{figure}
\centering
    \includegraphics[width=.5\linewidth]{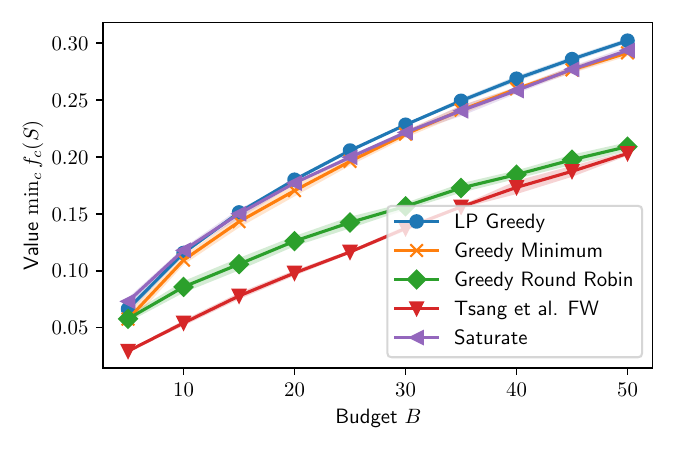}~
  \includegraphics[width=.5\linewidth]{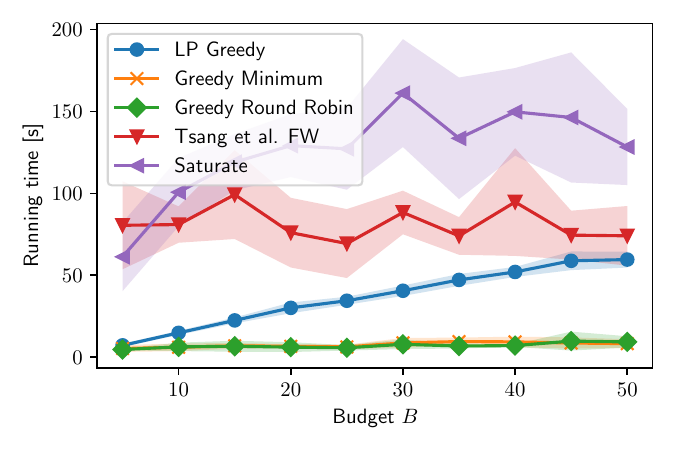}
  \caption{
  Fair influence maximization on a simulated Antelope Valley network of
  $n\!=\!500$ nodes on attribute \emph{ethnicity} with~$k\!=\!5$.
}
\label{fig:infmax-ethnicity}
\end{figure}

Finally, we replicate the setup of \citet{Tsang+19} for their
influence maximization objective, where the diffusion follows the independent
cascade model~\cite{Kempe+03}.
They define an instance for multiobjective submodular maximization via the colored influence
$f_c(S)\!=\!\frac 1 {|V_c|} \E[\textrm{number of nodes with color $c$ that $S$ activates}]$.
We use their simulated Antelope Valley networks on $n=500$ nodes which can be
colored according to different node attributes. As in \citet{Tsang+19}, we use $1\ 000$
samples to approximate the influence.
Figure~\ref{fig:infmax-ethnicity} shows results for the attribute
\emph{ethnicity}, where we outperform the prior work especially
for large budgets. We omit \textsc{Udwani MWU} as it
exceeds 10 minutes per run.
The algorithm of \citet{Tsang+19} uses special
evaluation oracles for the gradient of the influence which do not correspond
to function evaluations, so we report the running time instead.
Recall that their algorithm is specialized to influence maximization, for which
we can evaluate the multilinear extension efficiently.
Our algorithm applies to any multiobjective instance.
We show results for further attributes in Appendix~\ref{sec:appendix-fair-influence-max}, which includes
one instance where the heuristic \textsc{Saturate} outperforms our algorithm
for a small budget. Results for the remaining graphs in \citet{Tsang+19} are
similar, so we omit them.

\subsection{Ablation Study}

\begin{figure}
\centering
\includegraphics[width=0.5\linewidth]{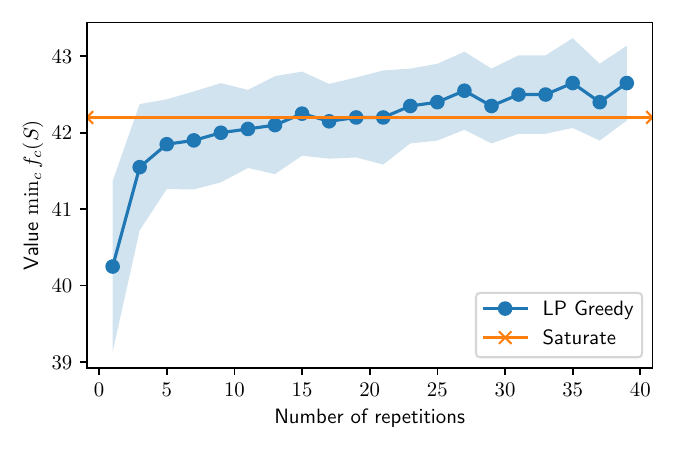}~
\includegraphics[width=0.5\linewidth]{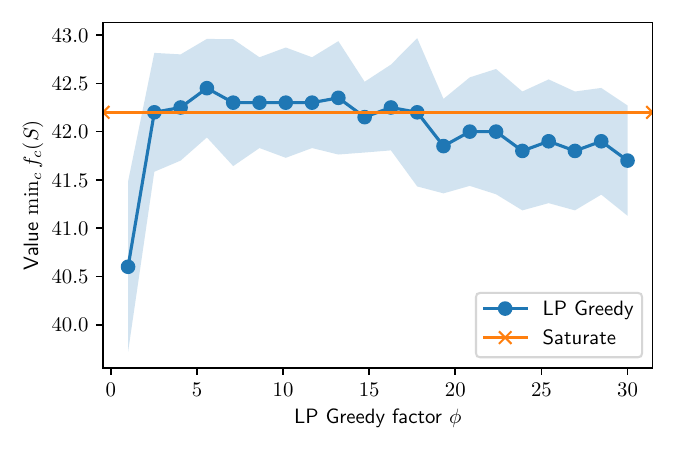}
\caption{
  We run Algorithm~\ref{alg:multiobjective} on a max-$k$-cover instance of
  $k=20$ Erd\H{o}s-R\'enyi random graphs on $n=64$ nodes with $p=0.1$.
  For the left plot, we vary the number of repetitions while using $\phi=10$.   For the right plot, we vary the factor $\phi$ while using $20$ repetitions.
  We report mean and standard deviation over 5 runs.
}
\label{fig:ablation}
\end{figure}

Figure~\ref{fig:ablation} provides an ablation study for the number of
repetitions and the factor $\phi$ which makes the algorithm act
less conservatively. We use a simple synthetic instance of Erd\H{o}s-R\'enyi random
graphs for the max-$k$-cover problem. Our results show that after $20$
repetitions, our algorithm has surpassed the objective value of
\textsc{Saturate} and more repetitions do only contribute to a slight increase
in objective value. The right plot shows that a factor of approximately
$\phi \in [5, 15]$ results in the highest objective value.

\section{Conclusion}\label{sec:conclusion}

We introduce the first scalable algorithm for multiobjective
submodular maximization that achieves a
$(1 - \frac 1 e - \epsilon)$-approximation ratio.
Going beyond
fair centrality maximization and influence maximization, many problems in
fairness naturally admit formulations as multiobjective problems. The high
scalability and theoretical guarantees of our algorithm make it a well-suited
option to address such problems.
Further, we avoided the common methodology of optimizing in the continuous space
and rounding the resulting solution.
Our novel techniques may lead to improvements in areas even beyond the
multiobjective problem.

\bibliographystyle{plainnat}
\bibliography{main}

\newpage
\appendix

\section{Omitted Details and Proofs}
\label{sec:appendix-omitted-proofs}

\subsection{Missing Proofs from Section~\ref{sec:analysis}}

\expectationisgood*

\begin{proof}
  This follows directly by our algorithm design and the definition of
  expectation:
  \begin{align*}
    \mathbb E[ f_c(v^{(i)} \mid S^{(i-1)}) \mid S^{(i-1)}] &=
    \sum_{v \in V} \Pr[v^{(i)} = v \mid S^{(i-1)}] f_c(v \mid S^{(i-1)}) \\
    &= \sum_{v \in V} x_v^{(i)} f_c(v \mid S^{(i-1)}) \\
    &\ge \frac 1 B \left( \OPT - f_c(S^{(i-1)}) \right) .
  \end{align*}
\end{proof}

\expectedvalue*

\begin{proof}
  We do the unrolling as in the standard greedy analysis, but in expectation
  over the result of Lemma~\ref{lem:expectation-is-good}. That is,
  Lemma~\ref{lem:expectation-is-good} is equivalent to
  \begin{align*}
    \E[f_c(S^{(i)}) - \OPT + \OPT - f_c(S^{(i-1)}) \mid S^{(i-1)}] &\ge
    \frac 1 B \left( \OPT - f_c(S^{(i-1)}) \right) \\
    \iff \E[\OPT - f_c(S^{(i)}) \mid S^{(i-1)}] &\le
    \left( 1 - \frac 1 B \right) \left( \OPT - f_c(S^{(i-1)}) \right) .
  \end{align*}
  Now, unrolling this yields
  \begin{align*}
    \E[\OPT - f_c(S^{(B)})] &=
    \E_{S^{(B-1)}}\left[\E_{v^{(B)}}[\OPT - f_c(S^{(B)}) \mid S^{(B-1)}]\right] \\
    &\le \left(1 - \frac 1 B\right)\E_{S^{(B-1)}}\left[\OPT - f_c(S^{(B-1)})\right] \\
    &\le \cdots \\
    &\le \left(1 - \frac 1 B\right)^B \OPT \\
    &\le \frac 1 e \OPT, 
  \end{align*}
  which is equivalent to $\E[f_c(S^{(B)})] \ge (1 - \frac 1 e) \OPT$. 
\end{proof}

\subsection{Details about Removing the Condition on $\OPT$}
\label{sec:appendix-removing-condition}

\begin{algorithm2e}[t]
\KwIn{Monotone and submodular set functions $f_c \colon 2^V \to \R$ for $c \in C$ and per-color budget $B'$}
$T \gets \emptyset$\;
\For{$c \in C$}{
  \For{$i = 1, 2, \dots, B'$}{
    Let $v = \arg\max_{v \in V} f_c(v \mid T)$ \;
    Update $T \gets T \cup \{v\}$ \;
  }
}
\Return{$T$}.
\caption{Pre-Processing for Algorithm~\ref{alg:multiobjective}}
\label{alg:pre-processing-no-opt}
\label{alg:pre-processing}
\end{algorithm2e}

Before going into the analysis of Algorithm~\ref{alg:pre-processing} we want
to again motivate our algorithm design.
Instances can violate the condition $\OPT \ge \frac 4 {\epsilon^2} M \log(2k)$
only in two ways. First, if the budget is small, it may happen that the optimum
solution is not much larger than $M$. For instance, in the extreme case in which
$B = 1$, we also have $\OPT = M$ and necessarily violate the condition for
$\OPT$. Second, even if we have a large budget, the optimum solution value may
still not be much larger than $M$ due to submodularity, even for a single
color.
However, we will show
that the latter case is not an issue, and we only require a sufficiently large
budget. This rests on the fact that cases in which $M$ is large compared to
$\OPT$, there are only few elements with large marginal gain.
We can find those elements in a pre-processing step and add them to our
solution before running Algorithm~\ref{alg:multiobjective}.
The remaining elements can only have low marginal gain, which then allows us to
obtain a better concentration.

Recall that $\tilde C = \{c \in C : f_c(T) < \OPT \}$ where $T$ is the output
of Algorithm~\ref{alg:pre-processing-no-opt}.
Recall also that we define a modified instance where
$\tilde f_c(A) = f_c(A \cup T)$ for all $c \in \tilde C$ and $A \subseteq V$.

\begin{lemma}
  \label{lem:pre-processing-no-opt}
  \label{lem:pre-processing}
  Algorithm~\ref{alg:pre-processing-no-opt} outputs a set $T$ of size
  $|T| \le k B'$. Furthermore,
  $\tilde f_c(v \mid \emptyset) \le \frac \OPT {B'}$ for all $c \in \tilde C$ and $v \in V$.
\end{lemma}

\begin{proof}
  By the definition of Algorithm~\ref{alg:pre-processing-no-opt}, we add at
  most $B'$ elements per color, meaning the set $T$ has size at most $k B'$.
  To show the second claim,
  let now $c \in \tilde C$ and $v \in V$ be any element.
  If $c \in T$ then $\tilde f_c(v \mid \emptyset) = f_c(v \mid T) = 0$ and we
  are done. We may thus assume that we did not add $c$ to $T$.
  For the proof, let us now denote with $T_c$ the elements that are added to $T$
  for color $c$, and let $T_c^{(i)}$ denote $T_c$ at the end of iteration $i$.
  Since we did not add $v$ to $T_c$, the greedy selection implies that
  $f_c(v \mid T) \le f_c(v \mid T^{(i-1)}) \le f_c(v^{(i)} \mid T^{(i-1)})$
  for all $1 \le i \le B'$. Hence,
  \[
    f_c(v \mid T)
    \le \frac 1 {B'} \sum_{i=1}^{B'} f_c(v^{(i)} \mid T^{(i-1)})
    = \frac 1 {B'} f_c(T^{(B')} \mid T^{(0)})
    \le \frac 1 {B'} f_c(T)
    \le \frac 1 {B'} \OPT
  \]
  where the penultimate inequality is by submodularity, and the last inequality
  due to $c \in \tilde C$.
\end{proof}

In the end, we want to show that $T \cup \tilde S$ attains a good fraction
of the optimum solution. Clearly, Lemma~\ref{lem:pre-processing} shows that
we can effectively reduce the maximum singleton marginal gain which is
sufficient to improve the concentration in Theorem~\ref{thm:multiobjective}.
However, Algorithm~\ref{alg:multiobjective-simple} now runs with a reduced
budget $\tilde B < B$. We thus need to show that even with a reduced budget, we
can still get a good fraction of $\OPT$ as long as $\tilde B$ is sufficiently
large.

\optchange*

\begin{proof}
  We use a probabilistic argument.
  As before, we abuse notation and use $\widetilde \OPT_b$ to denote both
  the optimum solution and its value.
  Let $x^* \in \{0,1\}^V$ be such that
  $x^*_v = 1$ if $v \in \widetilde \OPT_B$ and $x^*_v = 0$ otherwise. 
    We define $\tilde x = (1 - \frac{\tilde B} B) x \in [0,1]^V$ as a similar
  vector but with a reduced budget.
  Let $\tilde F_c \colon [0,1]^V \to \mathbb R$ be the multilinear extension
  of $\tilde f_c$ for each $c \in \tilde C$.
  Since $x^* \in \{0,1\}^V$ we have $\tilde f_c(\widetilde \OPT_B) = \tilde F_c(x^*)$.
  As argued in Lemma 3 of Udwani's work~\cite{Udwani18}, we have
  \begin{align}
    \label{eq:9}
    \tilde F_c(\tilde x) \ge \frac {\tilde B} B \tilde F_c(x^*)
    = \frac{\tilde B} B \tilde f_c(\widetilde \OPT_B), 
  \end{align}
  which is due to Jensen's inequality and since the multilinear extension
  is concave in positive direction. We now use swap rounding~\cite{Chekuri+10}
  to obtain a set $\tilde S \subseteq V$ from $\tilde x$ of size
  $|\tilde S| = \tilde B$.  Importantly,
  swap rounding is an oblivious rounding scheme, meaning it does not evaluate
  $\tilde f_c$.     As such, we do not create a rounded solution specific to
  $\tilde f_c$, but the guarantee on the rounded solution holds for all
  $c \in \tilde C$. We have the following guarantee~\cite{Chekuri+10} for a
  rounded solution $\tilde S$, for all $\delta > 0$,
  \begin{align}
    \label{eq:10}
    \Pr[\tilde f_c(\tilde S) \le (1 - \delta) \tilde F_c(\tilde x)]
    \le \exp\left( - \frac{\tilde F_c(\tilde x) \delta^2}{2 M_c} \right)
    \le \exp\left( - \frac{\tilde B \tilde f_c(\widetilde \OPT_B) \delta^2}{2 B \gamma \widetilde \OPT_B} \right)
    \le \exp\left( - \frac{\delta^2 \tilde B}{2 \gamma B} \right) .
  \end{align}
    For the second inequality, we used that
  $M_c = \max_{v \in V} \tilde f_c(v \mid \emptyset) \le \gamma \widetilde \OPT_B$
  and Inequality~\eqref{eq:9}. For the last inequality, we used that
  $\tilde f_c(\widetilde \OPT_B) \ge \widetilde \OPT_B$.
  We need that the RHS of Inequality~\eqref{eq:10} is less than $1 / k$ to
  guarantee that the rounded solution exceeds $\tilde F_c(\tilde x)$ for all
  $c \in \tilde C$. In particular, we choose
  \[
    \delta = \sqrt{3 \gamma \frac B {\tilde B} \log k} 
    \implies
    \exp\left( - \frac{\delta^2 \tilde B}{2 \gamma B} \right) < \frac 1 k .
  \]
    By a union bound over all $c \in C$,
  \[
    \Pr\left[\textrm{there exists a $c \in \tilde C$ with } \tilde f_c(\tilde S) \le
    (1 - \delta) \tilde F_c(\tilde x)\right] < k \cdot \frac 1 k = 1 .
  \]
    As such, the event that $\tilde S$ obtains value
  $\tilde f_c(\tilde S) \ge (1 - \delta) \tilde F_c(\tilde x)$ for all $c \in C$
  simultaneously has non-zero probability, meaning that such a set $\tilde S$
  exists. Let $\tilde S$ now be this set. We have for all $c \in \tilde C$ that
  \[
    \widetilde \OPT_{\tilde B} \ge
    \tilde f_c(\tilde S) \ge (1 - \delta) \tilde F_c(\tilde x)
    \ge (1 - \delta) \frac {\tilde B} B \widetilde \OPT_B .
  \]
\end{proof}

We run the pre-processing Algorithm~\ref{alg:pre-processing}
for a value $B > 0$ which we will specify later in
Theorem~\ref{thm:final}.
This outputs a set of colors $\tilde C \subseteq C$ and 
a partial solution $T \subseteq V$, and we run Algorithm~\ref{alg:multiobjective}
use objective functions defined as $\tilde f_c(A) = f_c(A \cup T)$ for $c \in \tilde C$ 
and the budget $\tilde B = B - |T|$.
Since for $c \in C \setminus \tilde C$ we have already
reached the optimum value, we simply ignore those colors and do not pass
them to Algorithm~\ref{alg:multiobjective}.
Algorithm~\ref{alg:multiobjective} outputs a set $\tilde S \subseteq V$
and the final solution of the combined algorithm is $T \cup \tilde S$.

We can now put everything together to get the guarantee on the combined
algorithm that only requires a large budget.

\final*

\begin{proof}
  We use $B = \frac 1 \gamma$ where $\gamma = \frac {\epsilon^2}{36 \log k}$ for
  Algorithm~\ref{alg:pre-processing} which yields a partial solution $T$
  and colors $\tilde C \subseteq C$.
  We show the theorem statement for all $c \in C$ separately.
  If $c \not\in \tilde C$, we have by monotonicity that
  \[
    f_c(T \cup \tilde S) \ge f_c(T) \ge \OPT .
  \]
  Let now $c \in \tilde C$.
  After the pre-processing, we have by
  Lemma~\ref{lem:pre-processing} that
  $|T| \le \frac k \gamma$ so $\tilde B = B - |T| \ge B - \frac k \gamma$ and
  hence, by Lemma~\ref{lem:opt-change},
  \begin{align}
    \label{eq:12}
    \widetilde \OPT_{\tilde B}
    \ge \widetilde \OPT_B \left( 1 - \frac k {B \gamma} \right) \left(1 -
      \sqrt{3 \gamma \frac{1}{1 - \frac k {B \gamma}} \log k} \right) .
  \end{align}
  Now, since $B \ge \frac {3 k} {\epsilon \gamma}$
  we get
  $\frac k {B \gamma} \le \frac \epsilon 3$
  as well as
  \[
    3 \gamma \frac{1}{1 - \frac k {B \gamma}} \log k
    \le 3 \gamma \frac{1}{1 - \frac \epsilon 3} \log k
    < 4 \gamma \log k
    = \frac{\epsilon^2}{9}, 
  \]
  which means we can bound \eqref{eq:12} further and obtain
  \begin{align}
    \label{eq:13}
    \widetilde \OPT_{\tilde B} \ge \left(1 - \frac 2 3 \epsilon \right) \widetilde \OPT_B .
  \end{align}
  Lemma~\ref{lem:pre-processing} also says that after pre-processing,
  the maximum singleton marginal gain is
  \begin{align*}
    \tilde M_c
    = \max_{v \in V} \tilde f_c(v \mid \emptyset)
    \le \gamma \OPT
    \le \gamma \widetilde \OPT_B
    \le 2 \gamma \widetilde \OPT_{\tilde B}, 
  \end{align*}
  where the last inequality follows from \eqref{eq:13}. 
  Hence, by Theorem~\ref{thm:multiobjective},
  \begin{align}
    \label{eq:11}
    \tilde f_{c}(S) \ge \left( 1- \frac 1 e - \frac \epsilon 3 \right) \widetilde \OPT_{\tilde B} .
  \end{align}
  with probability at least $1 - \delta$.
  Let us condition on the case that Algorithm~\ref{alg:multiobjective} was
  successful and \eqref{eq:11} holds.
  In this case,
    putting everything together yields
  \begin{align*}
    f_c(T \cup \tilde S)
    &= \tilde f_c(\tilde S) & \textrm{(definition of $\tilde f_c$)} \\
    &\ge \left( 1 - \frac 1 e - \frac \epsilon 3 \right)
      \widetilde \OPT_{\tilde B} & \textrm{(by \eqref{eq:11})} \\
    &\ge \left( 1 - \frac 1 e - \frac \epsilon 3 \right)
      \left(1 - \frac 2 3 \epsilon \right) \widetilde \OPT_B & \textrm{(by \eqref{eq:13})} \\
    &\ge \left( 1 - \frac 1 e - \epsilon \right) \widetilde \OPT_B \\
    &\ge \left( 1 - \frac 1 e - \epsilon \right) \OPT . &
      \textrm{(monotonicity)}
  \end{align*}
\end{proof}

\subsection{Missing Details and Proofs from Section~\ref{sec:lazy}}
\label{sec:appendix-lazy}

\begin{algorithm2e}[t]
  \KwIn{Monotone and submodular set functions $f_c \colon 2^V \to \R$ for
  $c \in C$, partial solution $S \subseteq V$, upper bounds $g_c(v)$ for all
  $v \in V$ and $c \in C$ such that $g_c(v) \ge f_c(v \mid S)$, step size
  $\eta > 0$, and number of iterations $T$}
Initialize $y_c \gets 1 / |C|$ for all $c \in C$ and $\overline x_v \gets 0$ for all $v \in V$\;
\For{$t = 1, 2, \dots, T$}{
  Let $v^* \gets \bot$\;
  \For{$v \in V$ \emph{in order decreasing in} $\sum_{c \in C} y_c g_c(v)$}{
    \If{$v^* \not= \bot$ \emph{and} $\sum_{c \in C} y_c g_c(v) \le \sum_{c \in C} y_c g_c(v^*)$}
    {
      {\bf break}\;
    }
    If not yet computed, evaluate $f_c(v \mid S)$ for all $c \in C$\;
    Update $g_c(v) \gets f_c(v \mid S)$ for all $c \in C$\;
    \If{$v^* = \bot$ \emph{or} $\sum_{c \in C} y_c g_c(v) > \sum_{c \in C} y_c g_c(v^*)$}
    {
      Set $v^* \gets v$\;
    }
  }
  Update $y_c \gets y_c (1 - \eta \ell_c(v^*)) = y_c (1 - \eta (B f_c(v^* \mid S) + f_c(S)))$\;
  Normalize $y_c \gets y_c / \|y\|_1$ \;
  Add $\overline x_{v^*} \gets \overline x_{v^*} + \frac 1 T$\;
}
\Return{$\overline x$ \emph{and the updated upper bounds} $g_c(v)$ \emph{for all} $v \in V$ \emph{and} $c \in C$.}
\caption{Solving the LP with MWU and Lazy Evaluations}
\label{alg:mwu}
\end{algorithm2e}

Let us now describe and motivate the design of Algorithm~\ref{alg:mwu}.
Imagine a single iteration of Algorithm~\ref{alg:multiobjective-simple} where
we have a partial solution $S = S^{(i-1)}$.
We can formulate the LP \eqref{eq:6} as a zero-sum game with
a payoff for each element $v \in V$ and each color $c \in C$ defined as
\begin{align}
  \label{eq:15}
  \ell_c(v) = B f_c(v \mid S) + f_c(S) .
\end{align}
We write our LP in terms of this payoff:
\[
  \mathrm{LP}(S) = \max_{x \in \Delta_V} \min_{c \in C} \left\{
    \sum_{v \in V} x_v \ell_c(v)
    \right\} = 
  \min_{y \in \Delta_C} \max_{v \in V} \left\{
    \sum_{c \in C} y_c \ell_c(v)
    \right\}
\]
where we are able to exchange minimum and maximum as this presents a zero-sum
game \cite{Neumann28}.
We use the min-max formulation where the first player plays a
distribution $y \in \Delta_C$ over colors and the second player responds
with the best response
\[
  v^* = \arg\max_{v \in V} \sum_{c\in C} y_c \ell_c(v)
  = \arg\max_{v \in V}\sum_{c \in C} y_c f_c(v \mid S)
\]
where the equality is true since $f_c(S)$ is constant in $v$.
Since the RHS is just a linear combination of marginal gains, we can now
employ lazy evaluations when finding the maximizer $v^*$: While maintaining
upper bounds $g_c(v) \ge f_c(v \mid S)$ for all $v \in V, c \in C$, it verify
that
$\sum_{c \in C} y_c f_c(v \mid S) \ge \sum_{c \in C} y_c g_c(v)$.
Since the solution $S$ we build increases and marginal gains are decreasing
due to submodularity, we use prior marginal gains for $g_c(v)$, and update
the upper bounds until we can show optimality of $v^*$.

Finally, we understand the colors as experts and, treating
$\ell(v) = (\ell_c(v))_{c \in C} \in \R^C$ as a loss vector, we update
$y$ via multiplicative weight updates.

We now want to bound the required number of iterations $T$. The MWU can only
guarantee an approximate solution to the LP with a distribution
$\overline x \in \Delta_V$, but this is sufficient for our purposes.
Indeed, the following holds:

\begin{lemma}
  \label{lem:mwu-apx}
  Assume that in every iteration $1 \le i \le B$ of
  Algorithm~\ref{alg:multiobjective-simple}, we use an approximate solution
  $\overline x^{(i)}$ with
  \[
    \sum_{v \in V} \overline x_v^{(i)} f_c(v \mid S^{(i)}) \ge
    \frac 1 B ((1 - \epsilon) \OPT - f_c(S^{(i)}))
  \]
  for all $v \in V$ and $S \subseteq V$. Then,
  Algorithm~\ref{alg:multiobjective-simple} still outputs a solution
  $S$ such that $f_c(S) \ge (1 - \frac 1 e - O(\epsilon)) \OPT$ for all
  colors $c \in C$, under the same conditions as
  Theorem~\ref{thm:multiobjective}.
\end{lemma}

\begin{proof}
  The proof is simple and follows, for example, by replacing
  $\OPT$ by $(1 - \epsilon) \OPT$ in the proof of
  Theorem~\ref{thm:multiobjective} and related lemmas.
\end{proof}

We therefore need to show that using sufficiently many iterations, the
multiplicative weights update can obtain an approximation $\overline x$
that satisfies the condition of Lemma~\ref{lem:mwu-apx}.

We use the typical regret guarantee (for example, see~\citet{Arora+12}) for
MWU, which states that
\begin{align}
  \label{eq:14}
  \min_{y \in \Delta_C} \sum_{t=1}^T \langle \ell(v^{(t)}), y \rangle
  \ge \sum_{t=1}^T \langle \ell(v^{(t)}), y^{(t)} \rangle -
    \frac{\log k}{\eta} - \eta \sum_{t=1}^T \|\ell(v^{(t)})\|_\infty^2 
\end{align}
where $y^{(t)}$ is the weight vector $y$ at the beginning of iteration $t$,
$\ell(v) = (\ell_c(v))_{c \in C} \in \R^C$ the loss vector as defined
in~\eqref{eq:15}, and $v^{(t)}$ the best
response to $y^{(t)}$.

\begin{restatable}{lemma}{mwuiter}
  \label{lem:mwu-iter}
  Using $T = 16 B^2 M^2 \frac{\log k}{\epsilon^2 \OPT^2}$ and an appropriate
  choice of $\eta>0$, the solution
  $\overline x \in \Delta_V$ is such that for all sets $S \subseteq V$
  with $|S| \le B$,
  \[
    \sum_{v \in V} \overline x_v f_c(v \mid S) \ge \frac 1 B \left(
    (1 - \epsilon) \OPT - f_c(S) \right) .
  \]
\end{restatable}

\begin{proof}
  Let $\overline y = \frac 1 T \sum_{t=1}^T y^{(t)} \in \Delta_C$.
  Note first that
  \[
    \frac 1 T \sum_{t=1}^T \langle \ell(v^{(t)}), y^{(t)} \rangle
    = \frac 1 T \sum_{t=1}^T \max_{v \in V} \langle \ell(v), y^{(t)} \rangle
    \ge \max_{v \in V} \frac 1 T \sum_{t=1}^T \langle \ell(v), y^{(t)} \rangle
    = \max_{v \in V} \langle \ell(v), \overline y \rangle
    \ge \OPT .
  \]
  As such, the regret guarantee~\eqref{eq:14} implies
  \begin{align}
    \label{eq:16}
    \min_{c \in C} \frac 1 T \sum_{t=1}^T \ell_c(v^{(t)}) =
    \min_{y \in \Delta_C} \frac 1 T \sum_{t=1}^T \langle \ell(v^{(t)}), y \rangle
    \ge \OPT -
    \frac{\log k}{T \eta} - \frac \eta T \sum_{t=1}^T \|\ell(v^{(t)})\|_\infty^2. 
  \end{align}
  Since
  $\ell_c(v^{(t)}) = B f_c(v \mid S) + f_c(S) \le B M + \OPT \le 2 B M$
  it suffices to set
  \[
    \eta = \frac{\epsilon}{8 B^2 M} \le \frac{\epsilon \OPT}{8 B^2 M^2} \quad \textrm{and} \quad
    T = 16 \frac{\log k B^2}{\epsilon^2} \ge 16 \frac{B^2 M^2 \log k}{\epsilon^2 \OPT^2}
      \]
  so we can further bound~\eqref{eq:16} by $(1 - \epsilon)\OPT$.
  Therefore, we have for all $c\in C$ that
  \[
    \sum_{v \in V} \overline x_v (B f_c(v \mid S) + f(S))
    = \frac 1 T \sum_{t=1}^T \ell_c(v^{(t)}) \ge (1 - \epsilon) \OPT = \frac{\log k}{\epsilon \eta \OPT}, 
  \]
  where the equality is
  by definition of $\overline x$ and the loss $\ell(v)$.
  Since $\overline x \in \Delta_V$,
  this is equivalent to the statement we wanted to show, so we are done.
\end{proof}

\runningtime*

\begin{proof}
  The pre-processing takes time $O(k n)$ since it has at most one function
  evaluation per color $c \in C$ and element $v \in V$.
          Algorithm~\ref{alg:multiobjective} runs through a total
  of $O(B \log(1 / \delta))$ iterations:
  $\log(2 / \delta)$ iterations of the outer loop and
  $B$ iterations of the inner loop.
  To solve the LP in each iteration,
  we need the marginal gains for all colors $c \in C$
  and elements $v \in V$, which thus requires $O(k n)$ function evaluations.
  We can solve the LP approximately which, as shown in
  Lemma~\ref{lem:mwu-iter}, requires
  $
    T = O(B^2 M \frac{\log k}{\epsilon^2 \OPT})
    = O(B^2 \frac 1 {\epsilon^2} \log k)
  $
  iterations.
              Each iteration of the multiplicative weights update involves at most $n k$
  function evaluations. Overall, we thus have a running time of
  $O(n B^3 \frac 1 {\epsilon^2} k \log(k) \log(1 / \delta))$ and
  $O(n B k)$ function evaluations.
\end{proof}

\section{Further Experimental Results}

\subsection{Max-$k$-Coverage}
\label{sec:appendix-max-coverage}

\begin{figure}
\centering
\includegraphics[width=0.5\linewidth]{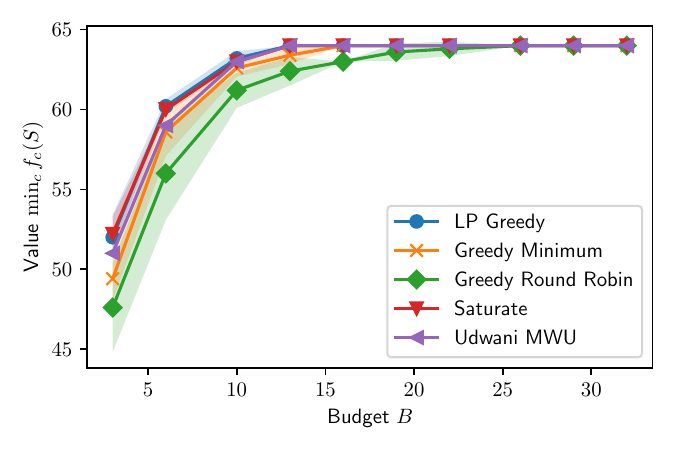}~
\includegraphics[width=0.5\linewidth]{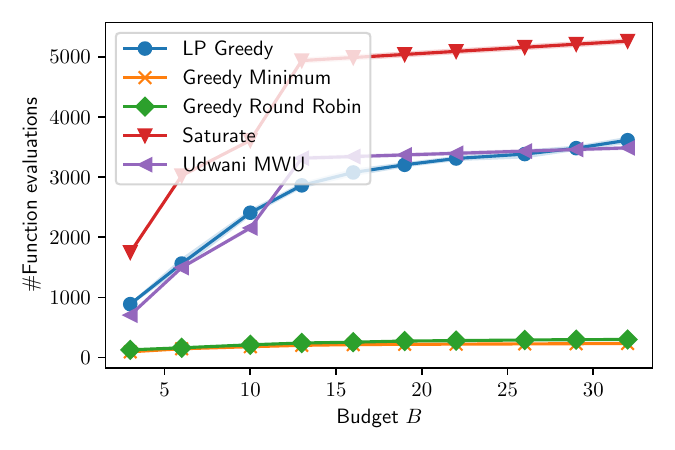}
\caption{
  Multiobjective submodular maximization for max-$k$-cover. We use $k=20$
  Barab\'asi-Albert graphs on $n=64$ nodes. We show the function value (top) and
  the number of evaluations (bottom). We report mean and standard deviation
  over 5 runs.
}
\label{fig:coverage-ba}
\end{figure}

\begin{figure}
\centering
\includegraphics[width=0.5\linewidth]{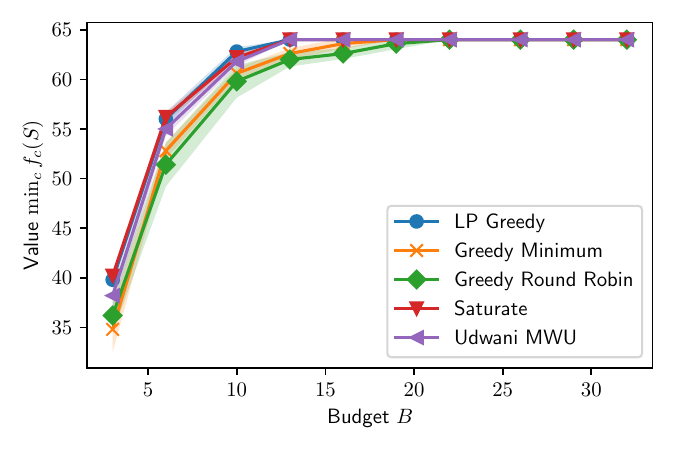}~
\includegraphics[width=0.5\linewidth]{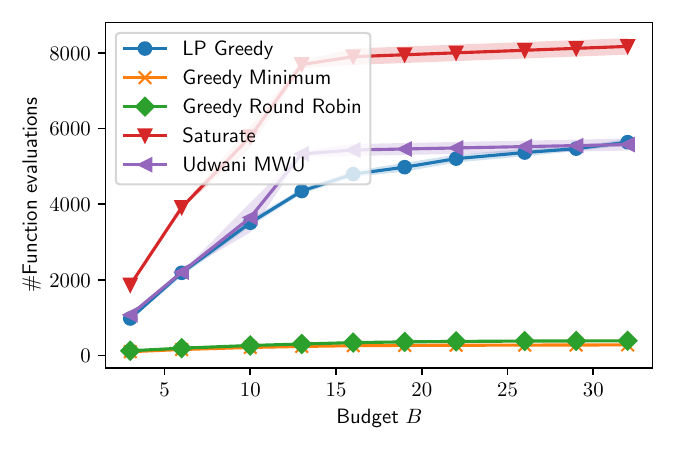}
\caption{
  Multiobjective submodular maximization for max-$k$-cover. We use $k=20$
  Erd\H{o}s-R\'enyi graphs on $n=64$ nodes. We show the function value (top) and
  the number of evaluations (bottom). We report mean and standard deviation
  over 5 runs.
}
\label{fig:coverage-gnp}
\end{figure}

\begin{figure}
\centering
\includegraphics[width=0.5\linewidth]{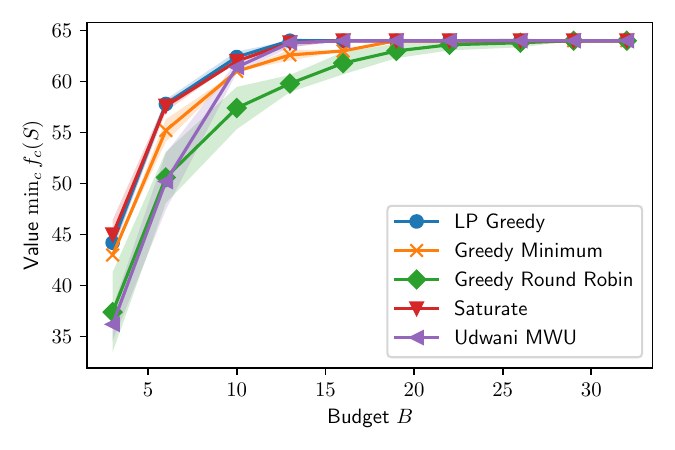}~
\includegraphics[width=0.5\linewidth]{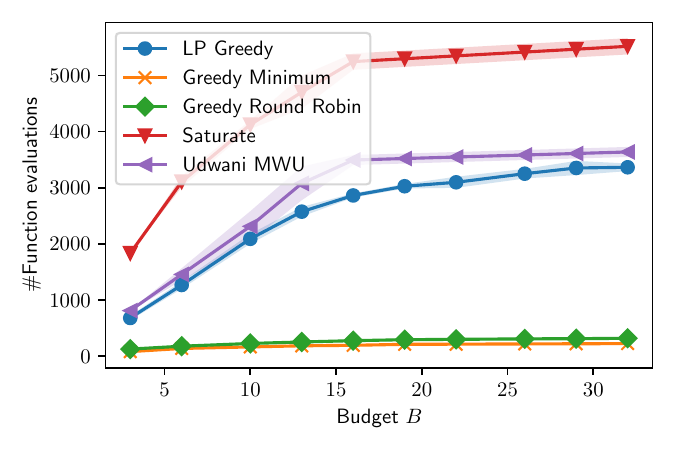}
\caption{
  Multiobjective submodular maximization for max-$k$-cover. We use $k=20$
  Barab\'asi-Albert graphs on $n=64$ nodes for varying $d$. We show the function
  value (top) and the number of evaluations (bottom). We report mean and
  standard deviation over 5 runs.
}
\label{fig:coverage-ba-hard}
\end{figure}

\begin{figure}
\centering
\includegraphics[width=0.5\linewidth]{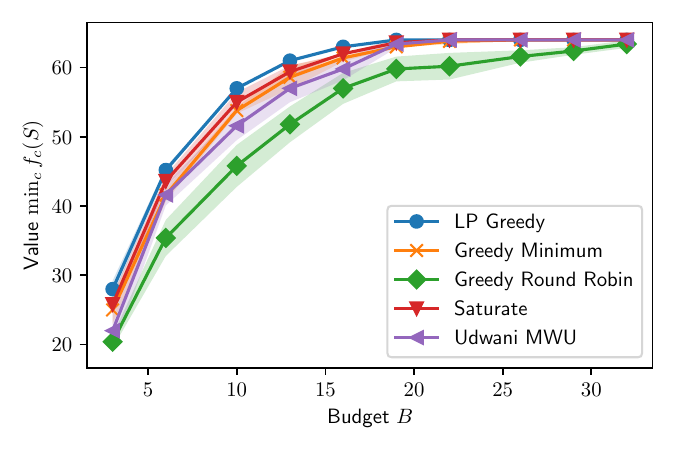}~
\includegraphics[width=0.5\linewidth]{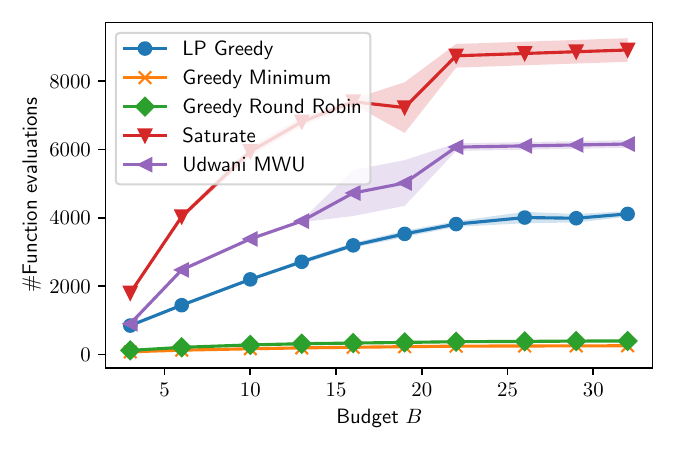}
\caption{
  Multiobjective submodular maximization for max-$k$-cover. We use $k=20$
  Erd\H{o}s-R\'enyi graphs on $n=64$ nodes for varying $p$. We show the function
  value (top) and the number of evaluations (bottom). We report mean and
  standard deviation over 5 runs.
}
\label{fig:coverage-gnp-hard}
\end{figure}

Figures~\ref{fig:coverage-ba} and \ref{fig:coverage-gnp} show our results for
Barab\'asi-Albert and Erd\H{o}s-R\'enyi random graphs, respectively. We obtain a
preferential-attachment graph in the Barab\'asi-Albert model by iteratively
connecting each node to $d=5$ existing nodes, with probability proportional to
their degrees. We obtain a random graph in the Erd\H{o}s-R\'enyi model by including
each pair as an edge with probability $0.1$.

Additionally, we create more difficult synthetic instances where the graph
properties differ per color. For Erd\H{o}s-R\'enyi random graphs, we create an
instance where we use $p_c = 0.1 + \frac c {50} \in [0.1, 0.5]$ for colors
$1 \le c \le 20$ to generate the $c$-th graph.
For Barab\'asi-Albert graphs, we vary the number of
$d_c = \lceil 5 + \frac c 2 \rfloor$ for colors $1 \le c \le 20$.
Our results in
Figures~\ref{fig:coverage-ba-hard} and \ref{fig:coverage-gnp-hard} show
that the our algorithm performs better under imbalance compared to heuristics
such as \textsc{Saturate}.

\subsection{Fair Centrality Maximization}
\label{sec:appendix-fair-centrality-max}

\begin{figure}
\centering
\includegraphics[width=0.5\linewidth]{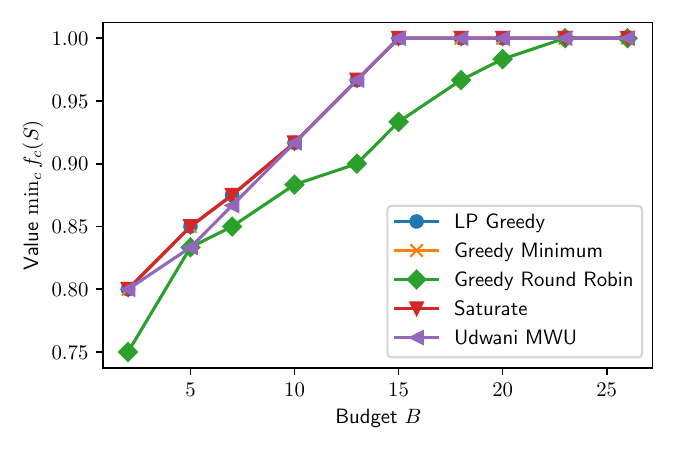}~
\includegraphics[width=0.5\linewidth]{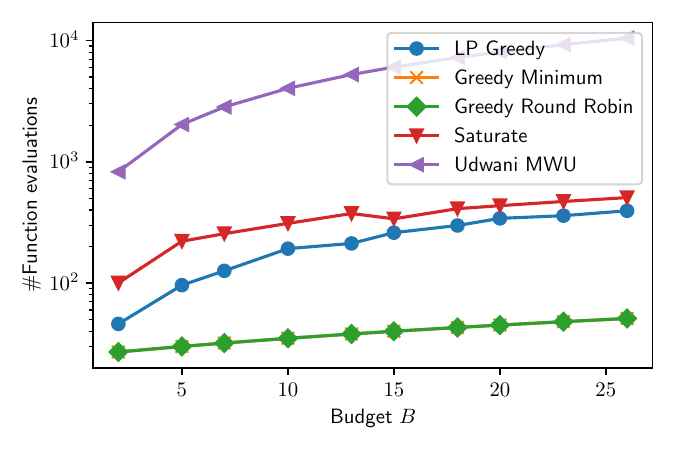}
\caption{
  Fair centrality maximization on the Amazon co-purchasing graph
  \emph{Movies \& TV} with $n=23$ nodes and $k=2$ colors.
}
\label{fig:fcm1}
\end{figure}

\begin{figure}
\centering
\includegraphics[width=0.5\linewidth]{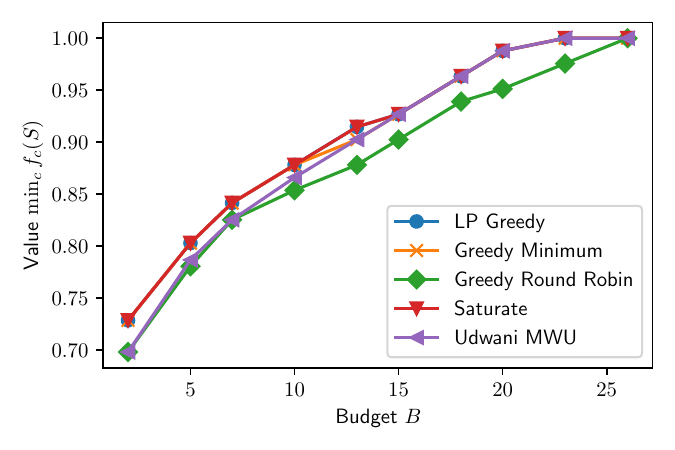}~
\includegraphics[width=0.5\linewidth]{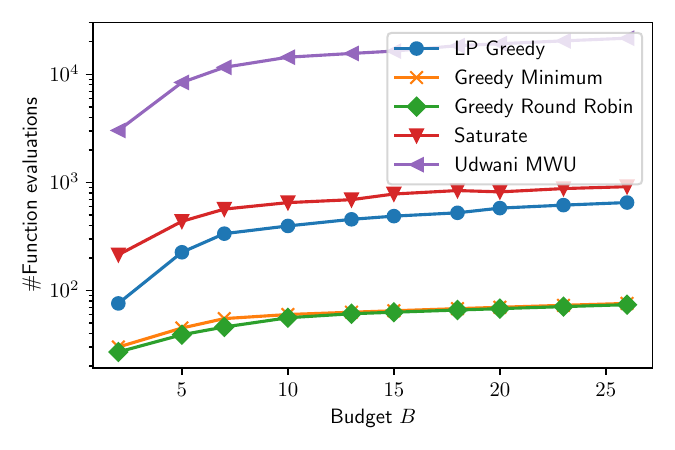}
\caption{
  Fair centrality maximization on the Amazon co-purchasing graph
  \emph{Musical Instruments} with $n=46$ nodes and $k=2$ colors.
}
\label{fig:fcm2}
\end{figure}

\begin{figure}
\centering
\includegraphics[width=0.5\linewidth]{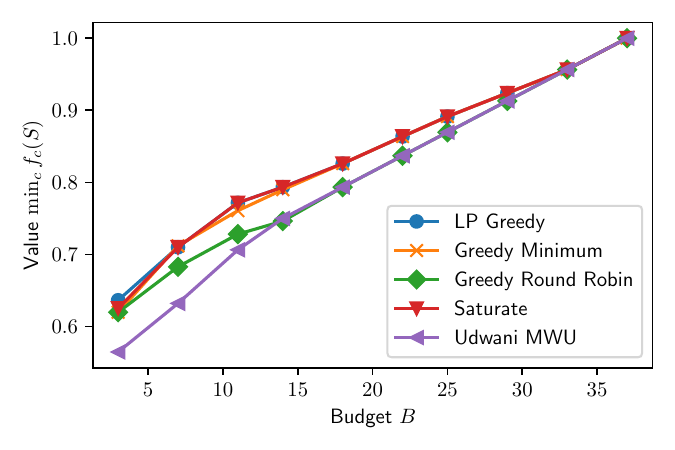}~
\includegraphics[width=0.5\linewidth]{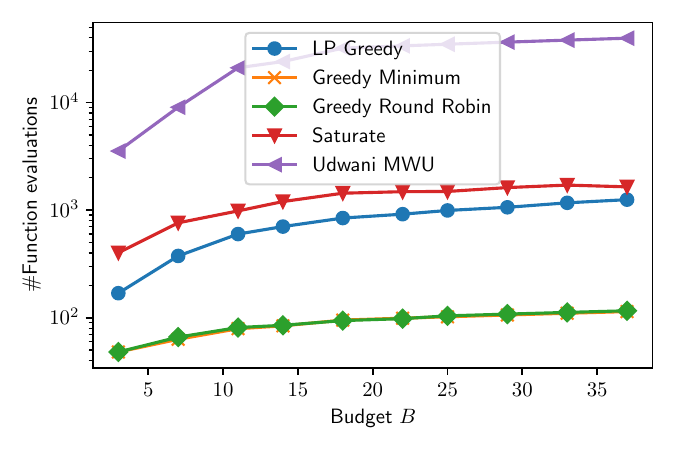}
\caption{
  Fair centrality maximization on the Amazon co-purchasing graph
  \emph{All Electronics} with $n=47$ nodes and $k=2$ colors.
}
\label{fig:fcm3}
\end{figure}

\begin{figure}
\centering
\includegraphics[width=0.5\linewidth]{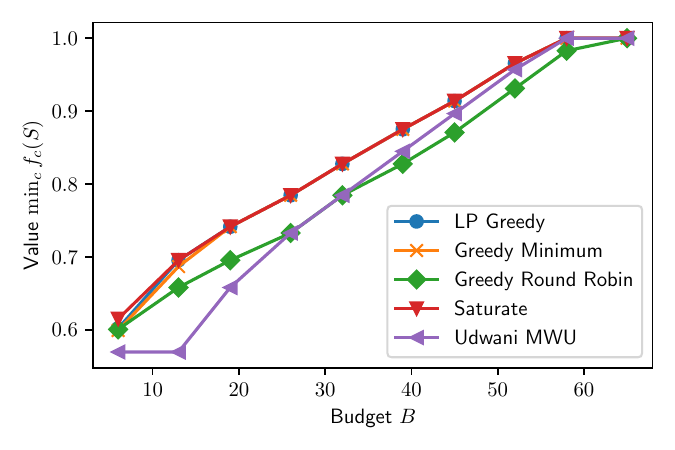}~
\includegraphics[width=0.5\linewidth]{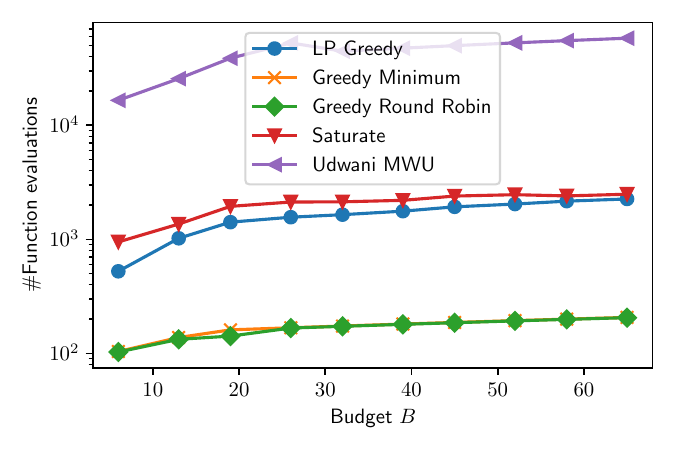}
\caption{
  Fair centrality maximization on the Amazon co-purchasing graph
  \emph{Computers} with $n=59$ nodes and $k=2$ colors.
}
\label{fig:fcm4}
\end{figure}

\begin{figure}
\centering
\includegraphics[width=0.5\linewidth]{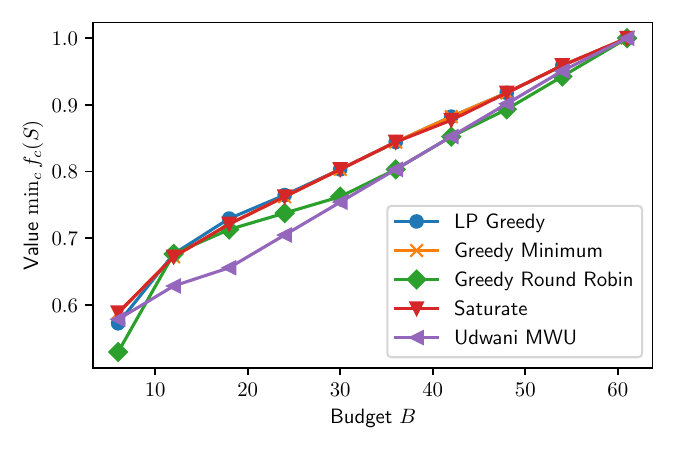}~
\includegraphics[width=0.5\linewidth]{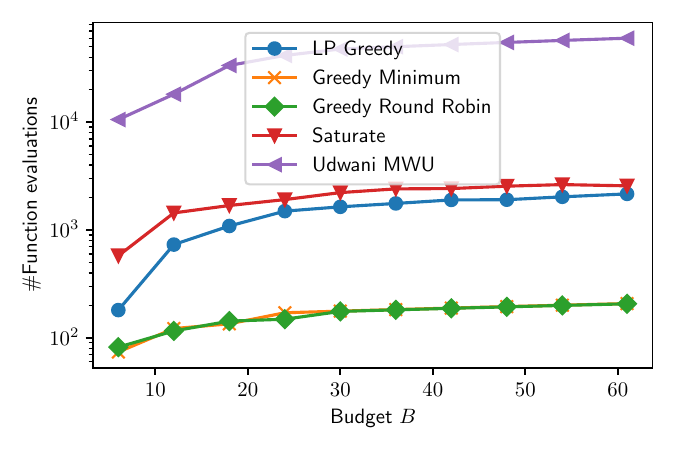}
\caption{
  Fair centrality maximization on the Amazon co-purchasing graph
  \emph{Home Audio \& Theater} with $n=77$ nodes and $k=2$ colors.
}
\label{fig:fcm5}
\end{figure}

\begin{figure}
\centering
\includegraphics[width=0.5\linewidth]{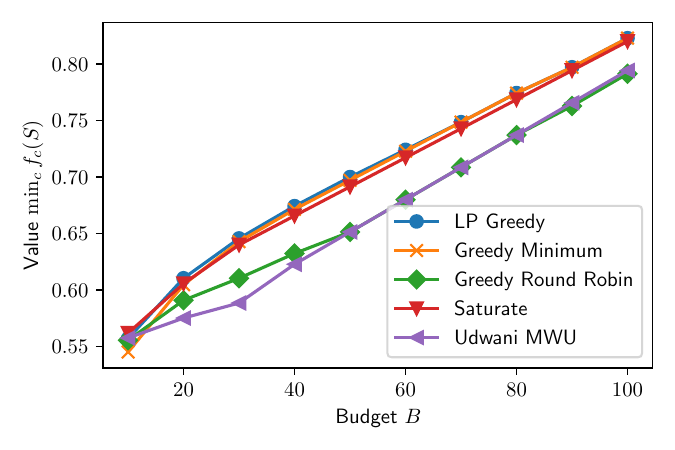}~
\includegraphics[width=0.5\linewidth]{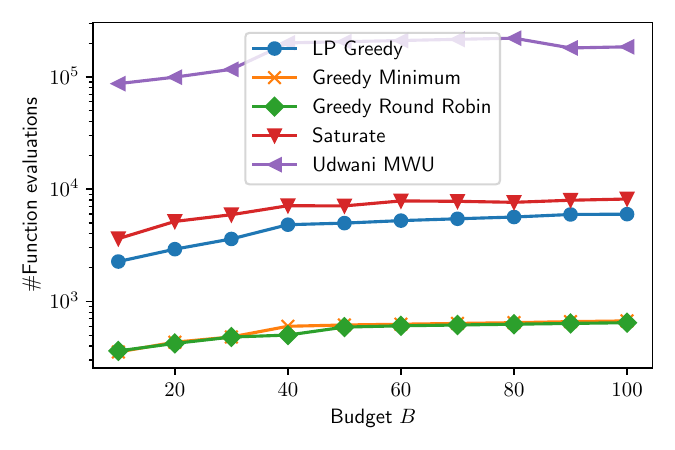}
\caption{
  Fair centrality maximization on the Amazon co-purchasing graph
  \emph{Camera \& Photo} with $n=202$ nodes and $k=2$ colors.
}
\label{fig:fcm6}
\end{figure}

\begin{figure}
\centering
\includegraphics[width=0.5\linewidth]{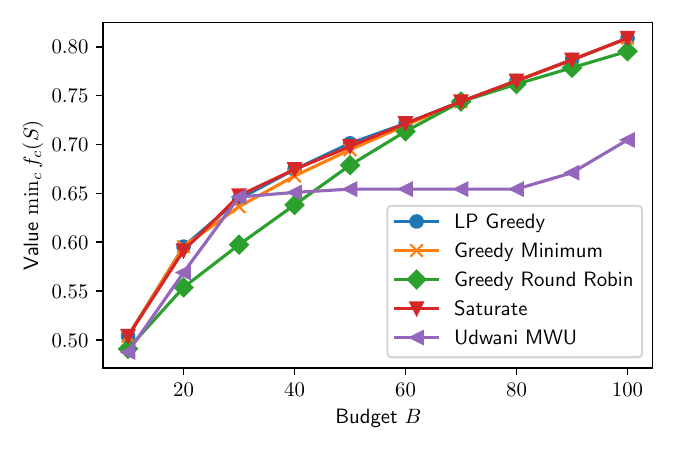}~
\includegraphics[width=0.5\linewidth]{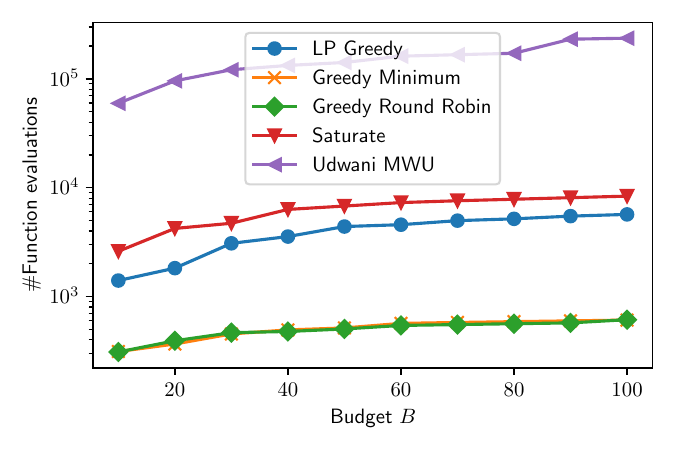}
\caption{
  Fair centrality maximization on the Amazon co-purchasing graph
  \emph{Baby} with $n=228$ nodes and $k=2$ colors.
}
\label{fig:fcm7}
\end{figure}

\begin{figure}
\centering
\includegraphics[width=0.5\linewidth]{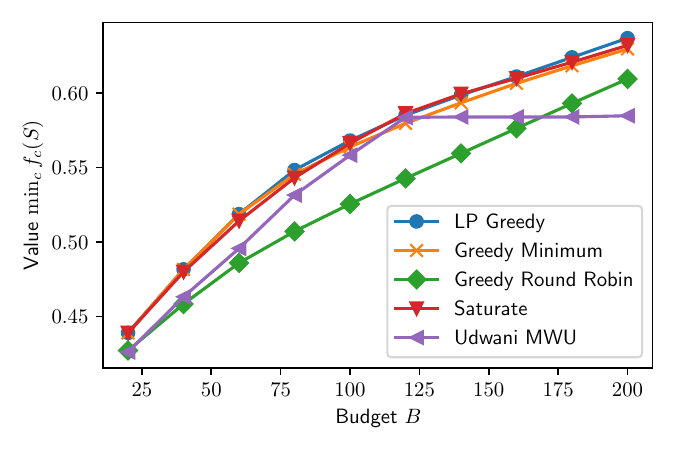}~
\includegraphics[width=0.5\linewidth]{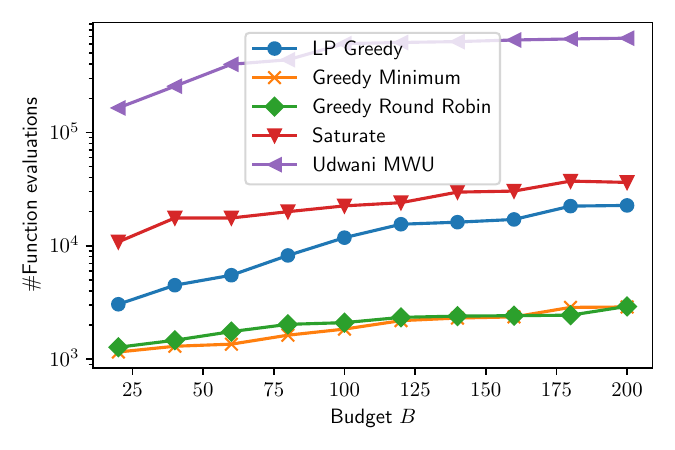}
\caption{
  Fair centrality maximization on the Amazon co-purchasing graph
  \emph{Luxury Beauty} with $n=1037$ nodes and $k=2$ colors.
}
\label{fig:fcm8}
\end{figure}

\begin{figure}
\centering
\includegraphics[width=0.5\linewidth]{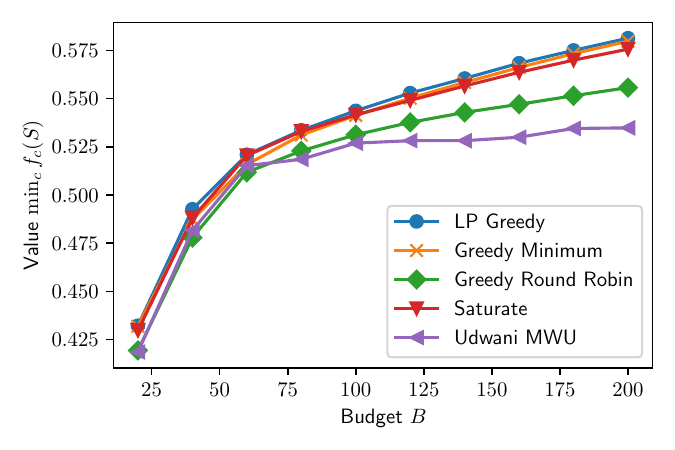}~
\includegraphics[width=0.5\linewidth]{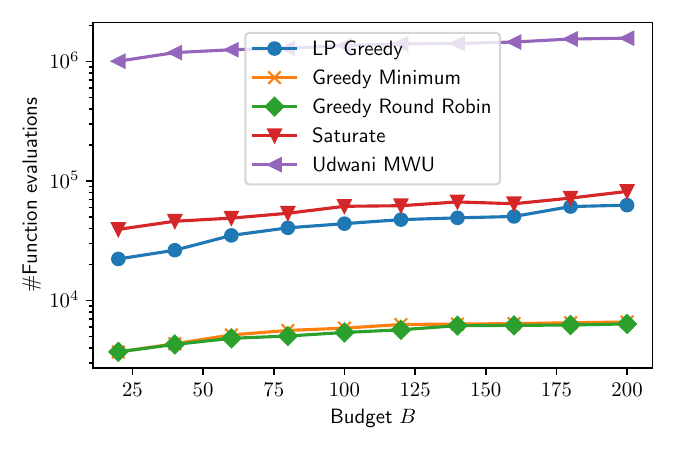}
\caption{
  Fair centrality maximization on the Amazon co-purchasing graph
  \emph{Pet Supplies} with $n=1785$ nodes and $k=2$ colors.
}
\label{fig:fcm9}
\end{figure}

\begin{figure}
\centering
\includegraphics[width=0.5\linewidth]{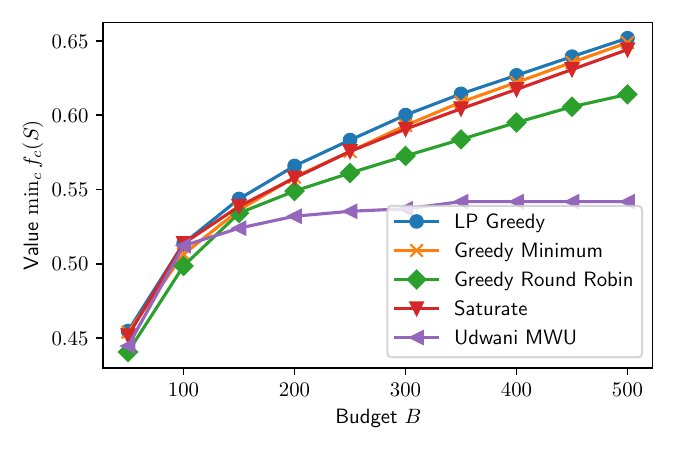}~
\includegraphics[width=0.5\linewidth]{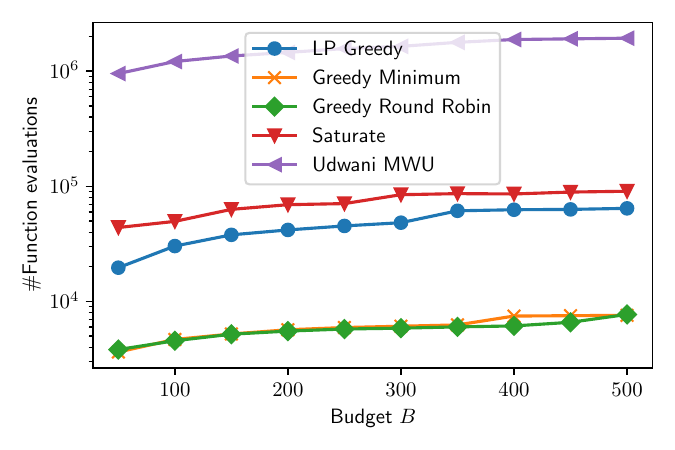}
\caption{
  Fair centrality maximization on the Amazon co-purchasing graph
  \emph{Industrial \& Scientific} with $n=2005$ nodes and $k=2$ colors.
}
\label{fig:fcm10}
\end{figure}

\begin{figure}
\centering
\includegraphics[width=0.5\linewidth]{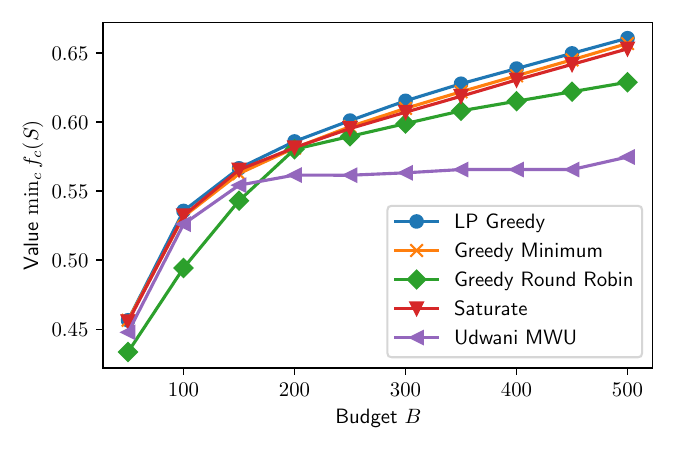}~
\includegraphics[width=0.5\linewidth]{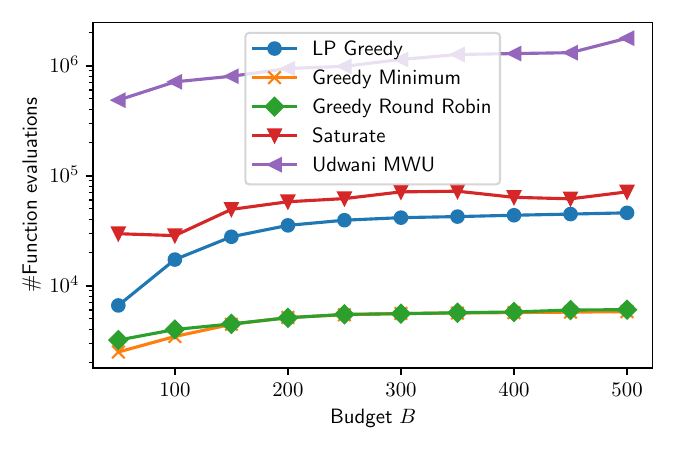}
\caption{
  Fair centrality maximization on the Amazon co-purchasing graph
  \emph{Office Products} with $n=2281$ nodes and $k=2$ colors.
}
\label{fig:fcm11}
\end{figure}

\begin{figure}
\centering
\includegraphics[width=0.5\linewidth]{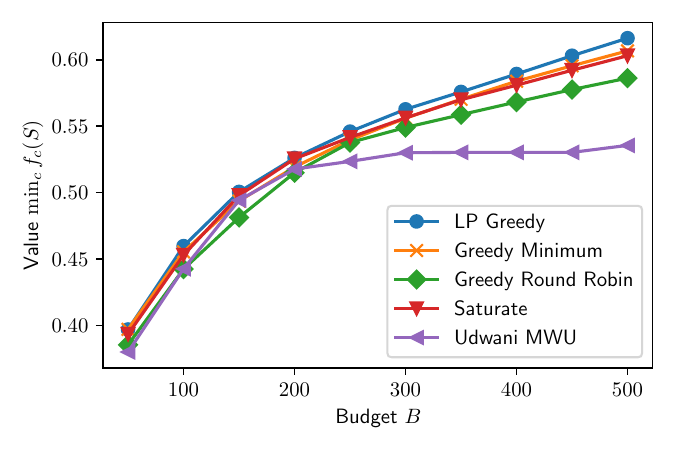}~
\includegraphics[width=0.5\linewidth]{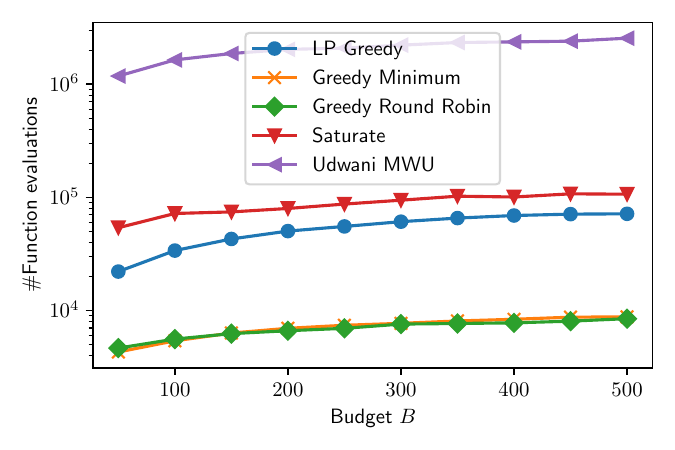}
\caption{
  Fair centrality maximization on the Amazon co-purchasing graph
  \emph{Books} with $n=2495$ nodes and $k=2$ colors.
}
\label{fig:fcm12}
\end{figure}

\begin{figure}
\centering
\includegraphics[width=0.5\linewidth]{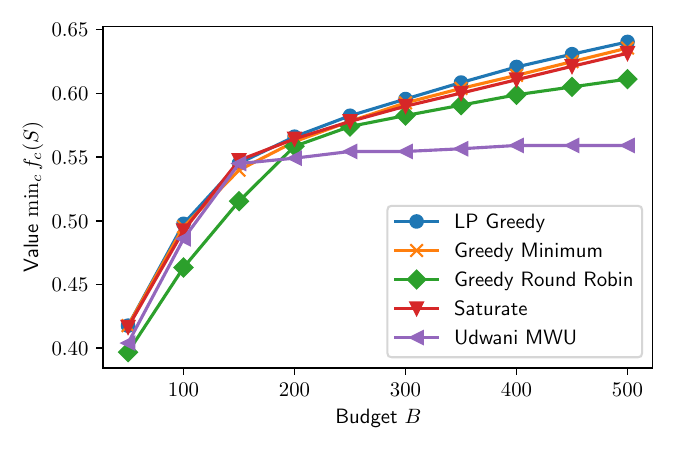}~
\includegraphics[width=0.5\linewidth]{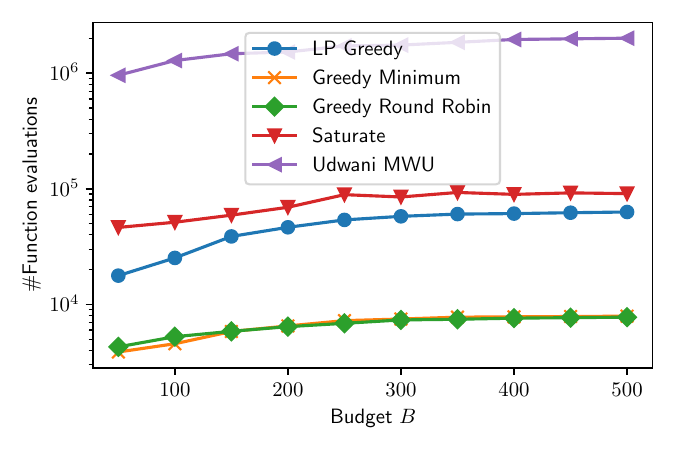}
\caption{
  Fair centrality maximization on the Amazon co-purchasing graph
  \emph{Home Improvement} with $n=2565$ nodes and $k=2$ colors.
}
\label{fig:fcm13}
\end{figure}

\begin{figure}
\centering
\includegraphics[width=0.5\linewidth]{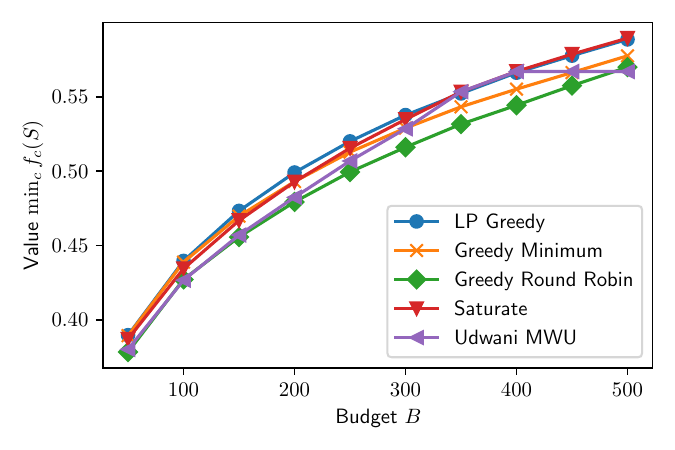}~
\includegraphics[width=0.5\linewidth]{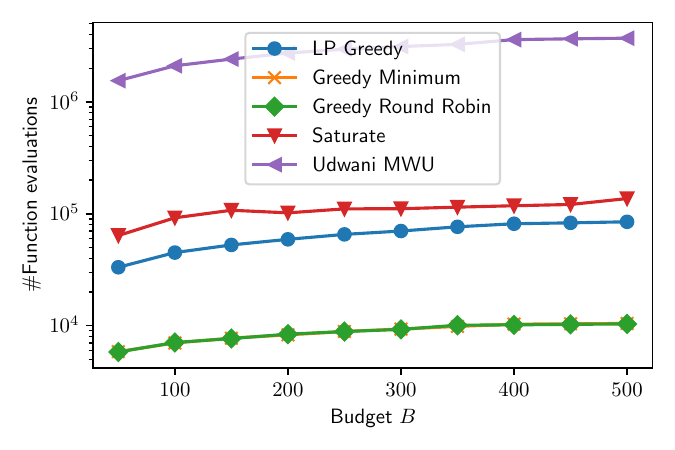}
\caption{
  Fair centrality maximization on the Amazon co-purchasing graph
  \emph{Health \& Personal Care} with $n=3010$ nodes and $k=2$ colors.
}
\label{fig:fcm14}
\end{figure}

\begin{figure}
\centering
\includegraphics[width=0.5\linewidth]{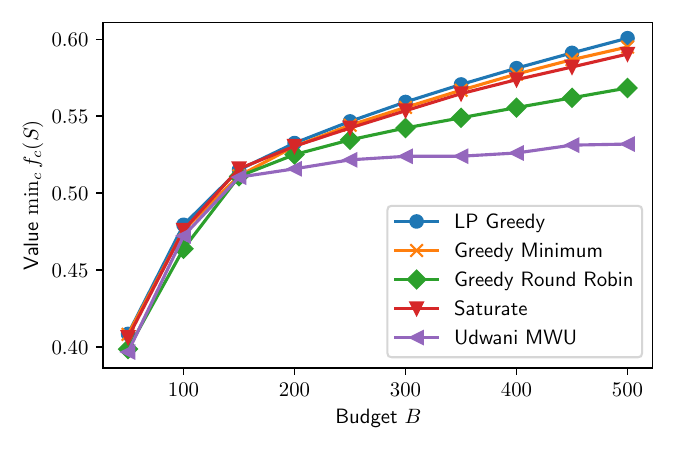}~
\includegraphics[width=0.5\linewidth]{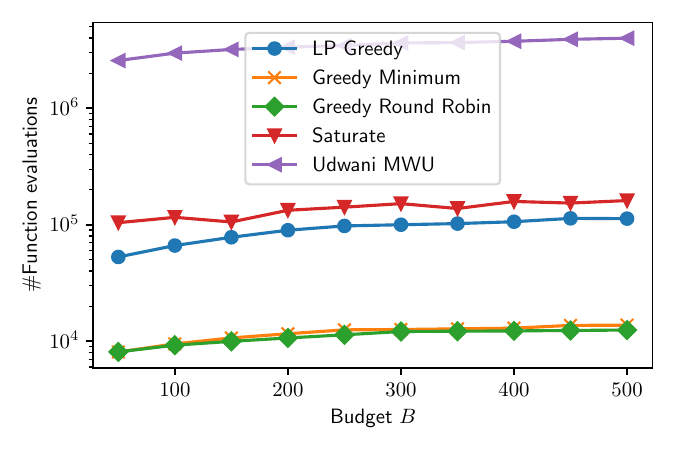}
\caption{
  Fair centrality maximization on the Amazon co-purchasing graph
  \emph{Sports \& Outdoors} with $n=3214$ nodes and $k=2$ colors.
}
\label{fig:fcm15}
\end{figure}

\begin{figure}
\centering
\includegraphics[width=0.5\linewidth]{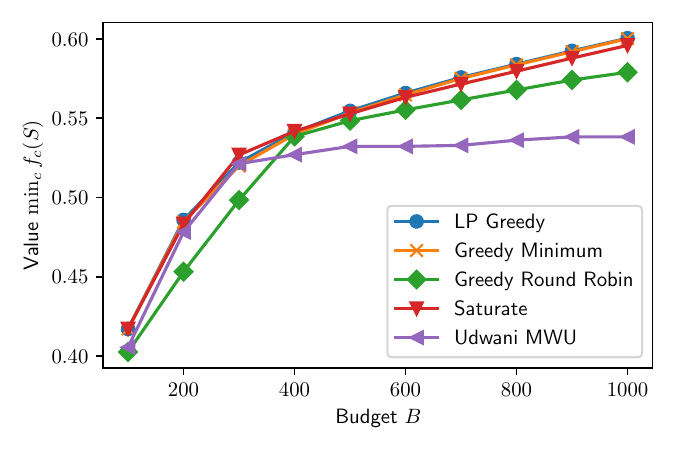}~
\includegraphics[width=0.5\linewidth]{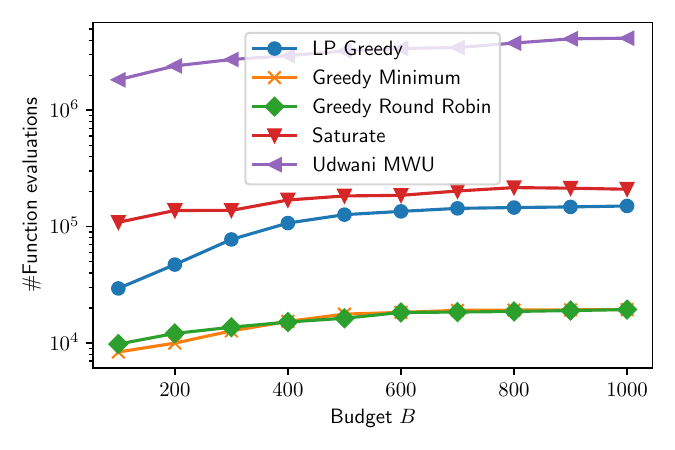}
\caption{
  Fair centrality maximization on the Amazon co-purchasing graph
  \emph{Grocery} with $n=6433$ nodes and $k=2$ colors.
}
\label{fig:fcm17}
\end{figure}

\begin{figure}
\centering
\includegraphics[width=0.5\linewidth]{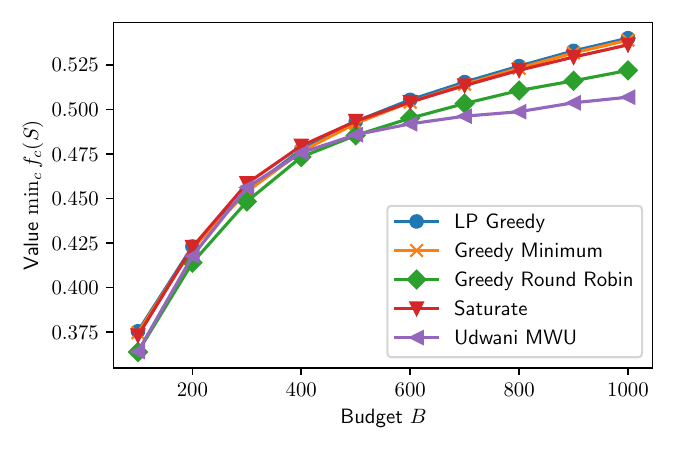}~
\includegraphics[width=0.5\linewidth]{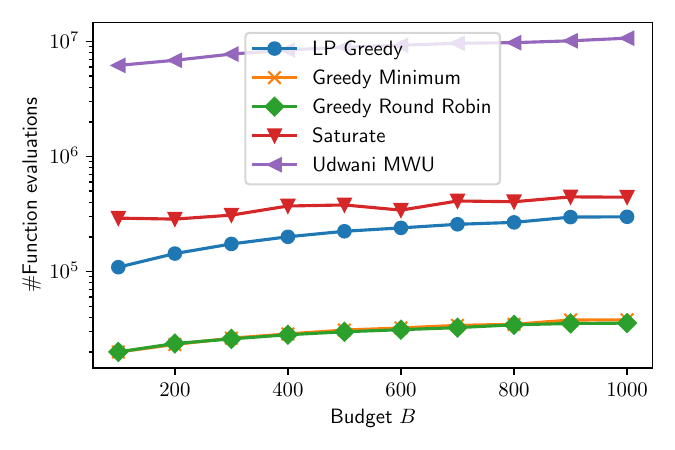}
\caption{
  Fair centrality maximization on the Amazon co-purchasing graph
  \emph{Amazon Home} with $n=10378$ nodes and $k=2$ colors.
}
\label{fig:fcm18}
\end{figure}

Figures~\ref{fig:fcm1} through \ref{fig:fcm18} show omitted results on Amazon
co-purchasing graphs.

\subsection{Fair Influence Maximization}
\label{sec:appendix-fair-influence-max}

\begin{figure}
\centering
  \includegraphics[width=0.5\linewidth]{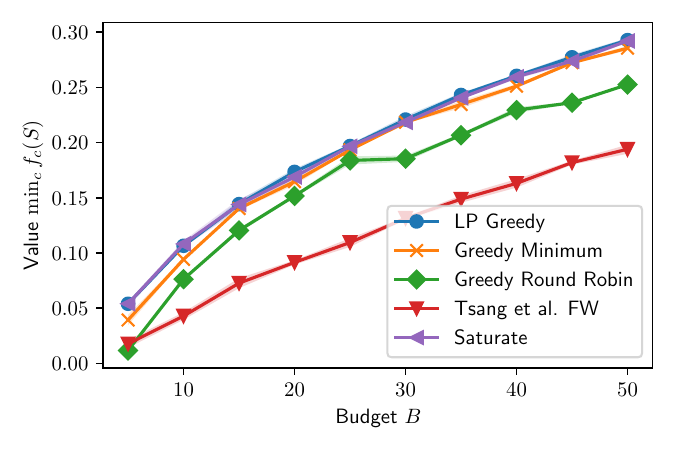}~
  \includegraphics[width=0.5\linewidth]{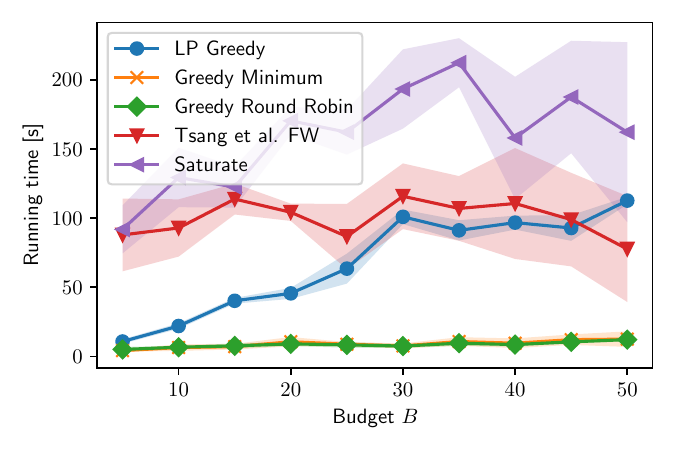}
\caption{
  Fair influence maximization on a simulated Antelope Valley network with attribute
  \emph{age} on $k=7$ colors.
}
\label{fig:infmax-age}
\end{figure}

\begin{figure}
\centering
  \includegraphics[width=0.5\linewidth]{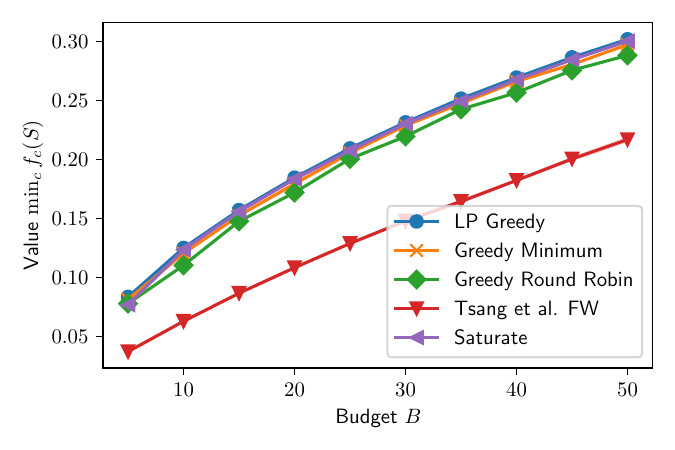}~
  \includegraphics[width=0.5\linewidth]{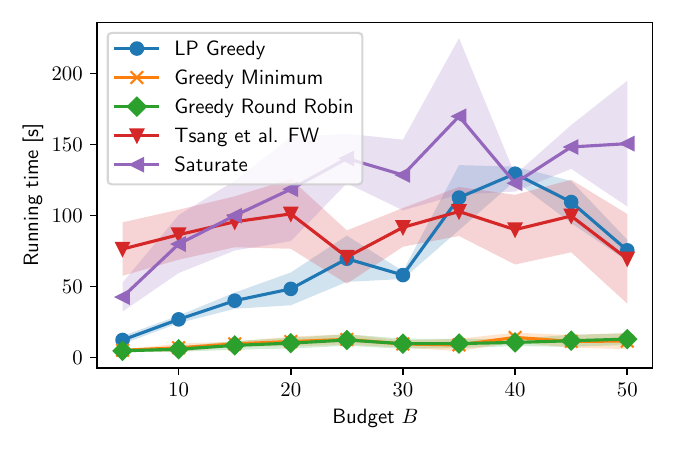}
\caption{
  Fair influence maximization on a simulated Antelope Valley network with attribute
  \emph{gender} on $k=2$ colors.
}
\label{fig:infmax-gender}
\end{figure}

\begin{figure}
\centering
  \includegraphics[width=0.5\linewidth]{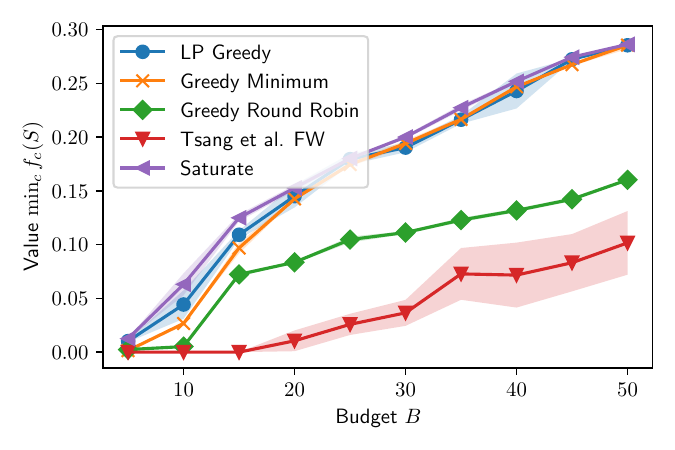}~
  \includegraphics[width=0.5\linewidth]{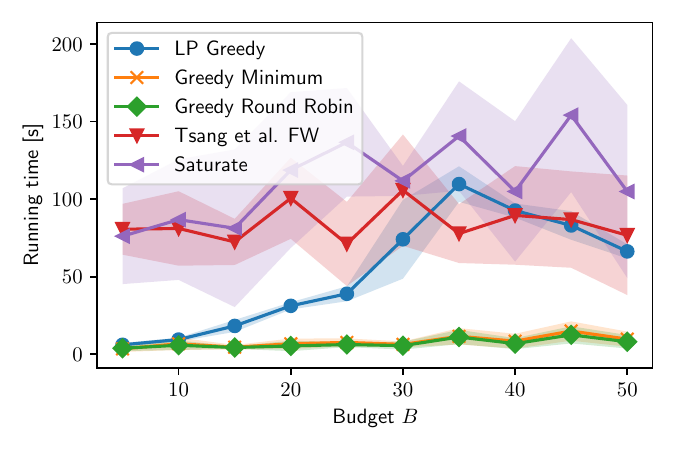}
\caption{
  Fair influence maximization on a simulated Antelope Valley network with attribute
  \emph{region} on $k=13$ colors.
}
\label{fig:infmax-region}
\end{figure}

\begin{figure}
\centering
  \includegraphics[width=0.5\linewidth]{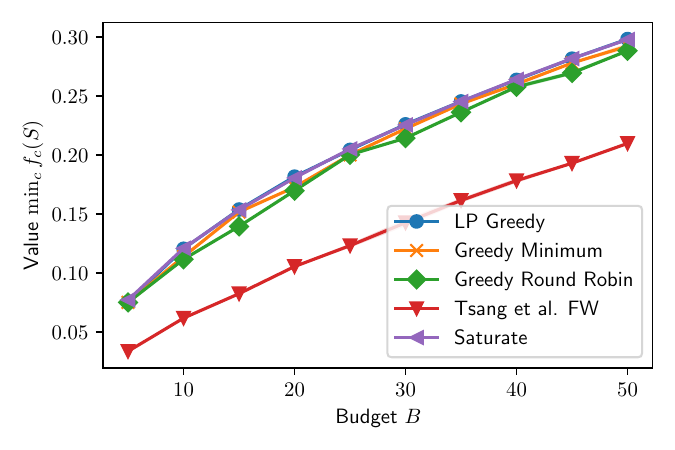}~
  \includegraphics[width=0.5\linewidth]{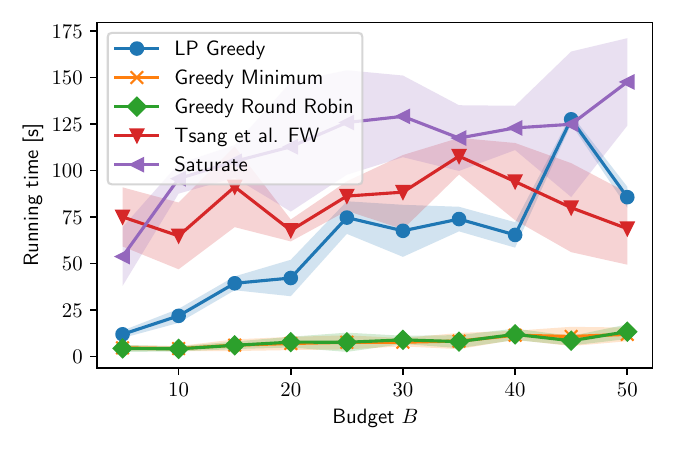}
\caption{
  Fair influence maximization on a simulated Antelope Valley network with attribute
  \emph{status} on $k=3$ colors.
}
\label{fig:infmax-status}
\end{figure}

Figures~\ref{fig:infmax-age} through \ref{fig:infmax-status} show omitted
results on an Antelope Valley network for four remaining attributes.

\end{document}